\documentclass{article}
\pdfoutput=1 

\usepackage{hyperref}
\usepackage{amsmath}
\usepackage{color}
\usepackage{cite}
\usepackage{amssymb}
\usepackage{amsthm}
\usepackage{graphicx}
\usepackage{xspace}
\usepackage{wrapfig}
\usepackage{algorithmic}
\usepackage{algorithm}
\usepackage{hyperref}
\usepackage{verbatim}
\usepackage{lineno}
\usepackage{anysize}
\marginsize{1in}{1in}{1in}{1in}

\newtheorem{lemma}{Lemma}
\newtheorem{theorem}{Theorem}

\newcommand{\Poly}{\ensuremath{\mathcal{P}}}               
\newcommand{\Vis}{\ensuremath{\mathrm{Vis}_\Poly(q)}} 
\newcommand{\Hout}{\ensuremath{\bar{h}}}      

\newcommand{\DS}{\ensuremath{{{\mathcal{DS}}}}}     

\newcommand{\Fir}{\ensuremath{{\overline{\mathcal{F}}}}}     
\newcommand{\Nex}{\ensuremath{{\overline{\mathcal{S}}}}}     
\newcommand{\BFir}{\ensuremath{{\mathcal{F}}}}     
\newcommand{\BNex}{\ensuremath{{\mathcal{S}}}}     

\newcommand{\unb}[1]{\ensuremath{\hat{a}_{#1}}}     
\newcommand{\lnb}[1]{\ensuremath{\check{a}_{#1}}}   

\newcommand{\alg}{\ensuremath{\mathcal{A}}\xspace} 
\newcommand{\lp}{\textsc{StackNeighbors}} 

\newcommand{\marrow}{\marginpar[\hfill$\longrightarrow$]{$\longleftarrow$}}

\newcommand{\blueremark}[2]{\textcolor{blue}{\textsc{#1 says:} \marrow\textsf{#2}}}

\newcommand{\luis}[1]{\blueremark{Luis}{#1}}

\newcommand{\ShoLong}[2]{#2} 

\graphicspath{{figures/}}

\title{Space-Time Trade-offs for Stack-Based Algorithms\footnote{A preliminary version of this paper appeared in the proceedings of the 30th Symposium on Theoretical Aspects of Computer Science (STACS 2013)~\cite{bklss-sttosba-12}.
}}

\author{
Luis Barba\thanks{Carleton University, Ottawa, Canada} 
$^{,}$\thanks{Universit\'e Libre de Bruxelles (ULB), Brussels, Belgium. {\tt \{lbarbafl,stefan.langerman\}@ulb.ac.be}.}
\and Matias Korman \thanks{%
National Institute of Informatics, Tokyo, Japan. Email: {\tt \{korman,sada\}@nii.ac.jp}} 
$^{,}$\thanks{
JST, ERATO, Kawarabayashi Large Graph Project.
}
\and Stefan Langerman\footnotemark[3]
$^{,}$\thanks{Directeur de Recherches du FRS-FNRS.}
\and Kunihiko Sadakane\footnotemark[4]
$^{,}$\thanks{Supported in part by JSPS KAKENHI 23240002.}
\and Rodrigo I. Silveira\thanks{Dept. de  Matem\'atica, Universidade de Aveiro, Portugal. {\tt{rodrigo.silveira@ua.pt}.}} $^{ }$  $^{,}$\thanks{Universitat Polit\`{e}cnica de Catalunya (UPC), Barcelona, Spain. {\tt  rodrigo.silveira@upc.edu}. }
}

\begin{document}
\maketitle

\begin{abstract}

In memory-constrained algorithms, access to the input is restricted to be read-only, and the number of extra variables that the algorithm can use is bounded.
In this paper we introduce the \emph{compressed stack technique}, a method that allows to transform algorithms whose main memory consumption takes the form of a stack into memory-constrained algorithms. 

Given an algorithm \alg\ that runs in $O(n)$ time using a stack of length $\Theta(n)$, we can modify it so that it runs in  $O(n^2\log n/2^s)$ time  using a workspace of $O(s)$ variables (for any $s\in o(\log n)$) or $O(n^{1+1/\log p})$ time using $O(p\log_p n)$ variables (for any $2\leq p\leq n$). We also show how the technique can be applied to solve various geometric problems, namely computing the convex hull of a simple polygon, a triangulation of a monotone polygon, the shortest path between two points inside a monotone polygon, a 1-dimensional pyramid approximation of a 1-dimensional vector, and the visibility profile of a point inside a simple polygon. 

Our approach improves or matches up to a $O(\log n)$ factor the running time of the best-known results for these problems in constant-workspace models (when they exist), and gives a trade-off between the size of the workspace and running time. To the best of our knowledge, this is the first general framework for obtaining memory-constrained algorithms.
\end{abstract}
\ShoLong{}{}

\section{Introduction}

The amount of resources available to computers is continuing to grow exponentially year after year.
Many algorithms are nowadays developed with little or no regard to the amount of memory used. However, with the appearance of specialized devices, there has been a renewed interest in algorithms that use as little memory as possible. 

Moreover, even if we can afford large amounts of memory, it might be preferable to limit the number of writing operations. For instance, writing into flash memory is a relatively slow and costly operation, which also reduces the lifetime of the memory. Write-access to removable memory devices might also be limited for technical or security reasons. Whenever several concurrent algorithms are working on the same data, write operations also become problematic due to concurrency problems. A possible way to deal with these situations is considering algorithms that do not modify the input, and use as few variables as possible. 

Several different  memory-constrained models exist in the literature. In most of them the input is considered to be in some kind of read-only data structure. In addition to the input, the algorithm is allowed to use a small amount of variables to solve the problem. In this paper, we look for space-time trade-off algorithms; that is, we devise algorithms that are allowed to use up to $s$ additional variables (for any parameter $s\leq n$). Naturally, our aim is that the running time of the algorithm decreases as $s$ grows. 

Many problems have been considered under this framework. In virtually all of the results, either an unconstrained algorithm is transformed to memory-constrained environments, or a new algorithm is created. Regardless of the type, the algorithm is usually an ad-hoc method tailored for the particular problem. In this paper, we take a different approach: we present a simple yet general approach to construct memory-constrained algorithms. Specifically, we present a method that transforms a class of algorithms whose space bottleneck is a stack into memory-constrained algorithms. In addition to being simple, our approach has the advantage of being able to work in a  {\em black-box} fashion: provided that some simple requirements are met, our technique can be applied to any stack-based algorithm without knowing specific details of their inner workings. 

\textbf{Stack Algorithms.}
One of the main algorithmic techniques in computational geometry is the incremental approach. At each step, a new element of the input is considered and some internal structure is updated in order to maintain a partial solution to the problem, which in the end will result in the final output. We focus on {\em stack algorithms}, that is, incremental algorithms where the internal structure is a stack (and possibly $O(1)$ extra variables). A more precise definition is given in Section~\ref{sec_prelim}.

We show how to transform any such algorithm into an algorithm that works in memory-constrained environments. The main idea behind our approach is to avoid storing the stack explicitly, reconstructing it whenever needed. The running time of our approach depends on the size of the workspace. Specifically, it runs in $O(n^{1+1/\log p})$ time and uses $O(p\log_p n)$ variables (for any $2\leq p\leq n$)~\footnote{In a preliminary version of this paper~\cite{bklss-sttosba-12}, we claimed slightly faster running times. Unfortunately, the claimed bounds turned out to be incorrect.}. In particular, when $p$ is a large constant (e.g. for $p=2^{1/\varepsilon}$) the algorithm runs in $O(n^{1+\varepsilon})$ time and uses $O(\log n)$ space. Also, if $p=n^\varepsilon$ the technique gives a linear-time algorithm that uses only $O(n^\varepsilon)$ variables (for any $\varepsilon>0$)

Our approach also works if only  $o(\log n)$ space is available, at the expense of restricting slightly the class of algorithms considered.
We say that a stack algorithm is {\em green\footnote{or environmentally friendly.}} if, without using the stack, it is possible to reconstruct certain parts of the stack efficiently (this will be formalized in Section~\ref{sec_green}).  We show how to transform any green stack algorithm into one that runs in $O(\frac{n^2 \log n}{2^s})$ time using $O(s)$ variables for any $s\in o(\log n)$. Afterwards, in Section~\ref{sec_hyb} we use the properties of green algorithms to obtain a speed-up for $\Omega(\log n)$ workspaces. A summary of the running time of our algorithms as a function of the available space (and the type of the algorithm) can be seen in Table~\ref{table_times}.

\begin{table}[ht] 
\centering 
\begin{tabular}{c c|| c |c ||c } 
 && \multicolumn{2}{c||}{ Running Time} & \\ \cline{3-4} 
\multicolumn{2}{c||}{Space} & Stack Alg. & Green Stack Alg. & Notes \\[0.9ex] \hline  
$O(s)$ &$\forall s\in o(\log n)$                                  & & $O(\frac{n^2\log n}{2^s})$ & Th.~\ref{theo_timelogn}\\
$O(\log n)$ &  $p=O(1)$                                                             & $O(n^{1+\varepsilon})$  & $O(n\log ^{2}n)$ & Th.~\ref{theo_gen}/Th.~\ref{theo_hybrid}\\
$O(n^{\frac{1}{(1+a)\log\log n}}\log\log n)$ & $0<a<1$& &$O(n\log^{1+a}n)$ &Th.~\ref{theo_hybrid} \\
$O(p\log_p n)$ & $p\in w(1) \cap o(n^{\varepsilon})$ & $O(n^{1+\frac{1}{\log p}})$& & Th.~\ref{theo_gen}\\
$O(n^{\varepsilon})$ & $p=n^{\varepsilon}$, $\forall\varepsilon\in (0,1)$ & $O(n)$ & $O(n)$ & Th.~\ref{theo_gen}\\
\hline 
\end{tabular} 
\caption{Complexity bounds of the several algorithms introduced in this paper} 
\label{table_times} 
\end{table} 

Our techniques are conceptually very simple, and can be used with any (green) stack algorithm in an essentially black-box fashion. We only need to replace the stack data structure with the compressed data structure explained, and create one or two additional operations for reconstructing elements in the stack. To the best of our knowledge, this is the first general framework for obtaining memory-constrained algorithms.

\textbf{Applications.}
The technique is applicable, among others, to the following well-known and fundamental geometric problems (illustrated in Fig.~\ref{fig:applications}):
(i) computing the convex hull of a simple polygon,
(ii) triangulating a monotone polygon,
(iii) computing a shortest path inside a simple polygon,
(iv) computing a 1-dimensional pyramid approximation of an histogram,
(v) computing the visibility polygon of a point in a simple polygon.
More details about these problems are presented in Sections~\ref{sec_prelim} and~\ref{sec_applis}.

\begin{figure}[tb]
\centering
\includegraphics[width=1\textwidth]{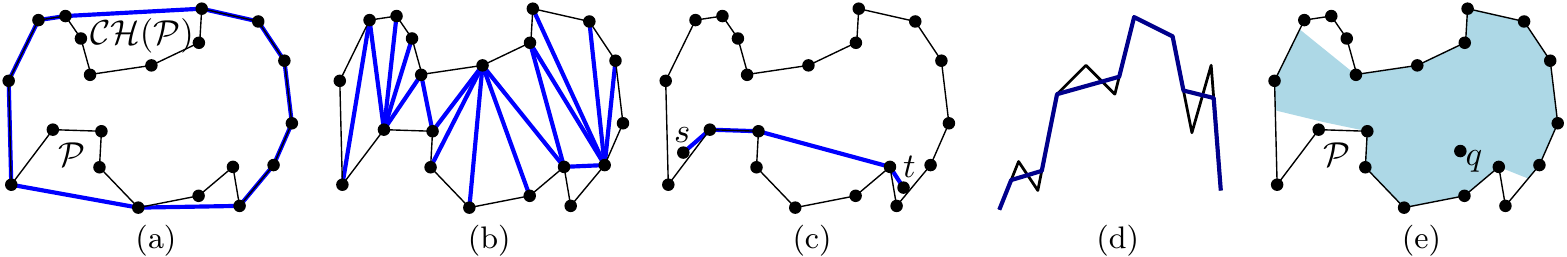}
\caption {Applications of the compressed stack, from left to right: convex hull of a simple polygon, triangulation of a monotone polygon, shortest path computation between two points inside a monotone polygon, optimal 1-d pyramid approximation, and visibility polygon of a point $q\in\Poly$.}
\label{fig:applications}
\end{figure}

We show in Section~\ref{sec_applis} that there exist green algorithms for all of the above applications except for the shortest path computation and pyramid approximation problems. Hence our technique results in new algorithms with time/space complexities as shown in Table~\ref{table_times}. In particular, when $p=n^{1/\varepsilon}$, they run in linear-time using $O(n^\varepsilon)$ variables, for any constant $\varepsilon>0$. 
%
%
For many of the applications discussed, our general trade-off matches or improves the best known algorithms throughout its space range.
For all the other applications, the running time of our technique differs from those of the fastest existing algorithms by at most a $O(\log n)$ factor.
%
A deeper comparison between the existing algorithms and our approach is given in Section~\ref{sec_similar}.

\textbf{Paper Organization} In Section~\ref{sec_prelim} we define our computational model, and we compare our results with those existing in the literature. Afterwards, in Section~\ref{sec_stack} we explain the compressed stack data structure. Essentially, our technique consists in replacing the $O(n)$-space stack of \alg\ by a \emph{compressed stack} that uses less space. This stack needs at least $\Omega(\log n)$ variables due to a hierarchical partition of the input.

For the case in which $o(\log n)$ space is available, we must add one requirement to the problem for our technique to work. We require the existence of a \lp\ operation that, given an input value $a\in \mathcal{I}$ and a consecutive interval $\mathcal{I}' \subseteq \mathcal{I}$ of the input, allows us to determine the elements of the input that were pushed into the stack immediately before and after $a$. Naturally, this operation must be efficient, and must not use too much space. Whenever such operation exists, we say that \alg\ is {\em green}. In Section~\ref{sec_green} we give a formal definition of green stack algorithms, and faster algorithms for them. In Section~\ref{seclargek} we explain how to handle the case in which \alg\ accesses the top $k$ elements of the stack (for any constant $k>1$). Finally, in Section \ref{sec_applis} we mention several applications of our technique, including examples of green algorithms.



\section{Preliminaries and related work}\label{sec_prelim}

\subsection{Memory-restricted models of computation}

Given its importance, a significant amount of research has focused on memory-constrained algorithms, some of them dating back to the 1980s~\cite{mp-ssls-80}. 
\ShoLong{}{One of the most studied problems in this setting is that of selection~\cite{mr-sromswmdm-96,Frederickson87,rr-iubtstsls-98,Chan}, where several time-space trade-off algorithms (and lower bounds) for different computation models exist.

}
In this paper, we use a generalization of the constant-workspace model, introduced by Asano {\em et al.}~\cite{amrw-cwagp-10,amw-cwspt-11}. In this model, the input of the problem is in a read-only data structure. In addition to the input, an algorithm can only use a constant number of additional variables to compute the solution. Implicit storage consumption by recursive calls is also considered part of the workspace. In complexity theory, the constant-workspace model has been studied under the name of {\em log space} algorithms~\cite{AB09}. In this paper, we are interested in allowing more than a constant number of workspace variables. Therefore, we say that an algorithm is an $s$-workspace algorithm if it uses a total of $O(s)$ variables during its execution. 
We aim for algorithms whose running times decrease as $s$ grows, effectively obtaining a space-time trade-off. Since the size of the output can be larger than our allowed space $s$, the solution is not stored but reported in a write-only memory. 
 
In the usual constant-workspace model, one is allowed to perform random access to any of the values of the input in constant time. The technique presented in this paper does not make use of such random access. Note that our technique will only use it when algorithm being adapted specifically needs it or in the case of green algorithms. Otherwise, our technique works in a more constrained model in which, given a pointer to a specific input value, we can either access it, copy it, or move the pointer to the previous or next input value.

Many other similar models in which the usage of memory is restricted exist in the literature. We note that in some of them (like the {\em streaming}~\cite{gk-seocqs-01} or the {\em multi-pass} model~\cite{cc-mpga-07}) the values of the input can only be read once or a fixed number of times. 
We follow the model of constant-workspace and allow scanning the input as many times as necessary. 
Another related topic is the study of \emph{succinct} data structures~\cite{JacoSucc89}. The aim of these structures is to use the minimum number of bits to represent a given input. Although this kind of approach drastically reduces the memory needed, in many cases $\Omega(n)$ bits are still necessary, and the input data structure needs to be modified.

Similarly, in the {\em in-place algorithm} model~\cite{bc-seacchspllt-06}, one is allowed a constant number of additional variables, but it is possible to rearrange (and sometimes even modify) the input values. This model is much more powerful, and allows to solve most of the above problems in the same time as unconstrained models. Thus, our model is particularly interesting when the input data cannot be modified, write operations are much more expensive than read operation, or whenever several programs need to access the same data concurrently. 

\subsection{Previous work on memory-constrained settings for our applications}
In this section we review previous work on memory-constrained algorithms for the different problems that can be solved with our technique.

\begin{description}

\item[Convex hull of a simple polygon] Given its importance, this problem has been intensively studied in memory-constrained environments. The well-known \emph{gift wrapping} or \emph{Jarvis march} algorithm~\cite{j-ich-73} fits our computation model, and reports the convex hull of a set of points (or a simple polygon) in $O(n\Hout)$ time using $O(1)$ variables, where \Hout\ is the number of vertices on the convex hull. 

Chan and Chen~\cite{cc-mpga-07} showed how to modify Graham's scan so as to compute the upper hull of set of $n$ points sorted in the $x$-coordinates.
Their algorithm runs in expected $O(n\log_p n)$ time using $O(p\log_p n)$ variables (see Theorem 4.5 of~\cite{cc-mpga-07}, and replace $n^\delta$ by $p$). 
Notice that points being sorted in the $x$-coordinates is equivalent to the fact that the algorithm works for $x$-monotone polygons (imagine virtually adding a point far enough to close the polygon). 

De {\em et al.}~\cite{dnr-chlproslws-12} used a similar strategy for the more difficult problem of computing the convex hull of a set of points (instead of a simple polygon) in memory-constrained workspaces. Their algorithm runs in $O(n^{1.5+\varepsilon})$ time using $O(\sqrt{n})$ space (for any $\sqrt{\frac{\log \log n}{\log n}}<\varepsilon<1$).  

Br\"onnimann and Chan~\cite{bc-seacchspllt-06} also modified the method of Lee~\cite{l-ofchsp-83} so as to obtain several linear-time algorithms using memory-reduced workspaces. However, their model of computation allows in-place rearranging (and sometimes modifying) the vertices of the input, and therefore does not fit into the memory-constrained model considered here. 


\item[Triangulation of a monotone polygon] The memory-constrained version of this problem was studied by Asano {\em et al.}~\cite{abbkmrs-mcasp-11}. In that paper, the authors give an algorithm that triangulates \emph{mountains} (a subclass of monotone polygons in which one of the chains is a segment). Combining this result with a trapezoidal decomposition, they give a method to triangulate a planar straight-line graph. Both operations run in quadratic-time in an $O(1)$-workspace. 

\item[Shortest path computation] Without memory restrictions, the shortest path between two points in a simple polygon can be computed in $O(n)$ time~\cite{gh-ospqsp-89}.
Asano {\em et al.}~\cite{amrw-cwagp-10} gave an $O(n^2)$ time algorithm for solving this problem with $O(1)$-workspace, which later was extended to $O(s)$-workspaces~\cite{abbkmrs-mcasp-11}. Their algorithm starts with a (possibly quadratic-time) preprocessing phase that consists in repeatedly triangulating $\Poly$, and storing $O(s)$ diagonals of $\Poly$ that partition it into $O(s)$ subpieces of size $O(n/s)$ each. Once the polygon is triangulated, they compute the geodesic between the two points in $O(n^2/s)$ time by navigating through the sub-polygons. Our triangulation algorithm removes the preprocessing overhead of Asano {\em et al.} when restricted to monotone polygons.

\item[Optimal 1-dimensional pyramid approximation] This problem is defined as follows: given an $n$-dimensional vector $f = (x_1,\ldots,x_n)$, find a unimodal vector $\phi = (y_1,\ldots,y_n)$ that minimizes the squared $L_2$-distance $|| f - \phi ||^2= \sum_{i=1}^n (x_i - y_i)^2$.
Linear time and space algorithms for the problem exist~\cite{cst-ltaacspc-06}, but up to now it had not been studied for memory-constrained settings.

\item[Visibility polygon (or profile) of a point in a simple polygon] This problem has been extensively studied for memory-constrained settings. Specifically, Asano {\em et al.}~\cite{amrw-cwagp-10} asked for a sub-quadratic algorithm for this problem in $O(1)$-workspaces. Barba {\em et al.}~\cite{bkls-cvpufv-11} provided a space-time trade-off algorithm that runs in $O(\frac{nr}{2^{s}}+n\log^2 r)$ time (or $O(\frac{nr}{2^{s}}+n\log r)$ randomized expected time)  using $O(s)$ variables (where $s\in O(\log r)$, and $r$ is the number of reflex vertices of $\Poly$). 

\end{description}

\subsection{Comparison with existing techniques}\label{sec_similar}

Most of the memory-constrained algorithms mentioned in the previous section use different techniques than ours, since each of them is designed to solve some specific problem. In this Section we discuss those that use similar techniques to the ones we use in this paper. 

Intuitively speaking, our approach is to partition the input into blocks, and run the algorithm as usual. In order to save space we  only store explicitly the information of the last two blocks that have been processed, while storing some small amount of information about the previous blocks. Whenever some information of older blocks is needed we {\em reconstruct} it. Reconstruction is obtained by re-executing the algorithm, but only on the missing portion.
The main contribution of our method is that it is general and does not depend on the problem being solved. As long as the algorithm is a stack algorithm, our technique can be applied. 


Among the existing algorithms that work in our memory-constrained model, the most similar result is due to De {\em et al.}~\cite{dmn-seavpsp-12}, who studied the memory-constrained version of computing the visibility polygon of a point in a simple polygon. Their approach is roughly the same as the one that we use in Section~\ref{sec_lin} for $O(\sqrt{n})$-workspaces.\footnote{This result was discovered in parallel and independently of our research.} Unfortunately, a recent article by Abrahamsen~\cite{a-aoaceevsp-13} claims that the algorithm of De {\em et al.} is incorrect; more details can be found in~\cite{a-cwavpp-uc-13}. These errors can probably be fixed by taking into account the winding of the polygon (as done in~\cite{js-clvpa-87}). 
Nevertheless, their approach only works for the computation of the visibility polygon.

In another recent paper, Barba et al.~\cite{bkls-cvpufv-11} also provided a time-space trade-off for the visibility polygon problem. 
The approach used there is based on partitioning the input into pieces and computing the visibility within each block independently. 
However, that algorithm uses specific geometric properties of visibility polygons, avoiding the need to reconstruct blocks, so in principle it cannot be extended to solve other problems. 

Chan and Chen~\cite{cc-mpga-07} used a similar idea for the computing the convex hull of a set of points sorted in the $x$ coordinates in memory-constrained workspaces. Using particular geometric properties of the input, they can show independence between blocks and avoiding the need for reconstruction. De {\em et al.}~\cite{dnr-chlproslws-12} also used the same idea for the general case. %
 In both cases, the computation of tangencies between blocks ends up becoming the time bottleneck of the algorithm. 



Finally, there exist several time-space trade-off methods for recognizing context-free grammars (CFGs)~\cite{c-dcfl-79,bcmv-rdcfl-83,rd-progmdcfl-00}. The time-space product for these results is roughly quadratic, which is not as good as what our methods give (for $\omega(1)$ workspaces). The standard algorithm for recognizing CFGs also uses a stack, thus one could think that our approach would improve their results. However, we note that context-free grammar recognition algorithms do not fall within the class of stack algorithms (for example: unlike in our model, a single element of the input could push $w(1)$ values into the stack). Thus, our approach cannot be used in their context (and {\em vice versa}). 

\subsection{Stack Algorithms}
Let \alg\ be a deterministic algorithm that uses a stack, and possibly other data structures \DS\ of total size $O(1)$. We assume that \alg\ uses a generalized stack structure that can access the last $k$ elements that have been pushed into the stack (for some constant $k$). That is, in addition to the standard \textsc{push} and \textsc{pop} operations, we can execute \textsc{top}$(i)$ to obtain the $i$-th topmost element (for consistency, this operation will return $\emptyset$ if either the stack does not have $i$ elements or $i>k$).

\begin{algorithm}
  \begin{algorithmic}[1]
    \STATE Initialize stack and auxiliary data structure \DS\ with $O(1)$ elements  from $\mathcal{I}$ 
    \FORALL{subsequent input $a \in \mathcal{I}$}
         \WHILE{some-condition($a$,\DS,\textsc{stack}.\textsc{top}(1),\ldots, \textsc{stack}.\textsc{top}($k$))} \label{algcondi}
          \STATE \textsc{stack}.\textsc{pop}
          \ENDWHILE        
          \IF{another-condition($a$,\DS,\textsc{stack}.\textsc{top}(1),\ldots, \textsc{stack}.\textsc{top}($k$))} \label{algocondi2}
          \STATE \textsc{stack}.\textsc{push}($a$) \label{algopush}
          \ENDIF  
    \ENDFOR
    \STATE Report(\textsc{stack})
  \end{algorithmic}
\caption{Basic scheme of a stack algorithm}
\label{alg:scheme}
\end{algorithm}

\ShoLong{} {
We consider algorithms \alg\ that have the structure shown in Algorithm~\ref{alg:scheme}. In such algorithms, the input is a list of elements $\mathcal{I}$, and the goal is to find a subset of $\mathcal{I}$ that satisfies some property. 
The algorithm solves the problem in an incremental fashion, scanning the elements of $\mathcal{I}$ one by one. At any point during the execution, the stack keeps the values that form the solution  up to now. When a new element $a$ is considered, the algorithm pops all values of the stack that do not satisfy certain condition, and if $a$ satisfies some other property, it is pushed into the stack. Then it proceeds with the next element, until all elements of $\mathcal{I}$ have been processed. The final result is normally contained in the stack, and at the end it is reported. This is done by simply reporting the top of the stack, popping the top vertex, and repeating until the stack is empty (imagine adding a virtual extra element ad the end of the input). }

We say that an algorithm that follows the scheme in Algorithm~\ref{alg:scheme} is a {\em stack} algorithm. Our scheme is mainly concerned on how the stack is handled. Thus, in principle one can make alterations to the scheme, provided that the treatment of the stack does not change and the extra operations do not affect the asymptotic running time of the algorithm. 
For simplicity of exposition, we assume that only values of the input are pushed, and that all values are different (see line~\ref{algopush} of Algorithm~\ref{alg:scheme}). 
In practice, this assumption can be removed by using the index $i$ of the value as identifier, and pushing any necessary additional information in a tuple whose identifier is $i$---as long as the tuple has $O(1)$-size. Note that we scan sequentially, hence the index is always known.


\section{Compressed stack technique}\label{sec_stack}
In this section we give a description of our compressed stack data structure, as well as the operations needed for the stack to work. First, we consider the case in which our workspace can use exactly $\Theta(\sqrt{n})$ variables. Afterwards, we generalize it to other workspace sizes.

\subsection{For $\Theta(\sqrt{n})$-workspaces}\label{sec_lin}

As a warm-up, we first show how to reduce the working space to $O(\sqrt{n})$ variables without increasing the asymptotic running time. For simplicity of exposition we describe our compressed stack technique for the case in which \alg\ only accesses the topmost element of the stack (that is, $k=1$). The general case for $k>1$ will be discussed in Section~\ref{seclargek}. 

Let $a_1, \ldots, a_n\in \mathcal{I}$ be the values of the input, in the order in which they are treated by \alg. 
\ShoLong{}{Although we do not need to explicitly know this ordering, recall that we assume that given $a_i$, we can access the next input element in constant time.} In order to avoid explicitly storing the stack, we virtually subdivide the values of $\mathcal{I}$ into blocks $B_1, \dots, B_{p}$, such that each block $B_i$ contains $n/p$ consecutive values\footnote{For simplicity of exposition, we assume that $p$ divides $n$. The general case will be considered in Section~\ref{sec_gen}.}. In this section we take $p=\sqrt n$. Then the size of each block will be $n/p=p=\sqrt{n}$. 
 Note that, since we scan the values of $\mathcal{I}$ in order, we always know to which block the current value belongs to.

We virtually group the elements in the stack according to the block that they belong to. Naturally, the stack can contain elements of different blocks. However, by the scheme of the algorithm, we know that all elements of one block will be  pushed consecutively into the stack. That is, if two elements of block $B_i$ are in the stack, all elements in the stack between them must also belong to block $B_i$. We call this the {\em monotone property} of the stack. Also, if at some point during the execution of $\alg$ there are elements from two blocks $B_i$ and $B_j$ in the stack (for some $i<j$), we know that those of $B_j$ must be closer to the top of the stack than those of $B_i$. We call this the {\em order} property. 


We introduce another important property of the stack, which we call the {\em invariant} property.
Suppose a value $a$ is in the stack at times $t$ and $t'$, $t<t'$, during the execution of \alg. 
Then at all times between $t$ and $t'$, $a$ is in the stack and the portion of the stack below $a$ is the same.
This follows from the fact that to remove $a$ (or some element below), we must first pop $a$. However, no value is pushed twice in a stack algorithm, thus it cannot afterwards be present in the stack at time $t'$. 

At any point of the execution, let $\BFir$ be the block that contains the element that was pushed last  into the stack. 
Note that it can happen that such element is not in the stack anymore (because it was popped by another value afterwards).
Similarly, let $\BNex$ be the topmost block with elements in the stack other than $\BFir$ (or $\BNex=\emptyset$ if no such block exists).
An important difference between $\BFir$ and  $\BNex$ is that
$\BFir$ may have no element currently in the stack, whereas  $\BNex$ always has elements in the stack, as long as $\BNex \neq \emptyset$.
%
%

Similarly, we use $\Fir$ and $\Nex$ to refer to the portions of the stack containing the elements of $\BFir$ and $\BNex$, respectively. 


We observe some properties of $\BFir$ and $\BNex$: by the order and monotonicity properties of the stack, the elements of $\Fir$ and $\Nex$ will form the top portion of the stack. Moreover, by definition of $\Nex$, the only situations in which $\Nex = \emptyset$ are when the stack contains elements from a single block or the stack is empty.

Since $\Fir$ and $\Nex$ contain values from two blocks, the portion of the stack that they define will not contain more than $2p = 2\sqrt{n}$ elements and thus, it can be stored explicitly. 
At any point during the execution we keep $\Fir$ and $\Nex$ in full. For any other block of the input, we only store the first and last elements of this block that are currently in the stack. We say that these blocks are stored in \emph{compressed} format. Note that we do not store anything for blocks that have no element currently into the stack. 

For any input value $a$, we define the {\em context} of $a$ as the content of the auxiliary data structure \DS\ 
 right after $a$ has been treated. Note that this context occupies $O(1)$ space in total. For each block that contains one or more elements in the stack (regardless of whether the block is stored explicitly or in compressed format) we also store the context of its first element that is in the stack. 

Therefore for most blocks we only have the topmost and bottommost elements that  are present in the stack, denoted $a_t$ and $a_b$, respectively, as well as the context of $a_b$. However, there could possibly be many more elements that we have not stored. For this reason, at some point during the execution of the algorithm we will need to reconstruct the missing elements in between. In order to do so, we introduce a \textsc{Reconstruct} operation. Given $a_t$, $a_b$ and the context of $a_b$,  \textsc{Reconstruct} explicitly recreates all elements between $a_b$ and $a_t$ that existed in the stack right after $a_t$ was processed. \ShoLong{}{A key property of our technique is that this reconstruction can be done efficiently:}

\begin{lemma}\label{lem_reconst}
\textsc{Reconstruct} runs in $O(m)$ time and uses $O(m)$ variables, where $m$ is the number of elements in the input between $a_b$ and $a_t$ in $\mathcal{I}$.
\end{lemma}
\begin{proof}
The way to implement \textsc{Reconstruct} is applying the same algorithm \alg, but starting from $a_b$, initialized with the context information stored with $a_b$, and stopping once we have processed $a_t$.  It suffices to show that running \alg\ in this partial way results in the same (partial) stack as running \alg\ for the whole input. The conditions evaluated in a stack algorithm (recall its structure in Algorithm~\ref{alg:scheme}) only use the context information: \DS\ and the top element of the stack. In particular, the results of local conditions tested during the execution of \alg\ between $a_b$ and $a_t$ will be equal to those executed during the execution of the \textsc{Reconstruct} procedure. 

Since \alg\ is a deterministic algorithm, the sequence of pops and pushes must be the same in both runs, hence after treating $a_t$ we will push $a_t$ in the stack (since it was pushed in the previous execution). By the invariant property of the stack, the elements in the stack below $a_t$ are preserved, and in particular we can explicitly reconstruct the portion of the stack between $a_t$ and $a_b$. Since there are at most $m$ input values between $a_b$ and $a_t$, the size of the stack during the reconstruction is bounded by $O(m)$. 
\end{proof}

Each time we invoke procedure  \textsc{Reconstruct} we do so with the first and last elements that were pushed into the stack of the corresponding block. In particular, we have the context of $a_b$ stored, hence we can correctly invoke the procedure. Also note that we have $m\leq n/p=p=\sqrt{n}$, hence this operation uses space within the desired bound. In order to obtain the desired running time, we must make sure that not too many unnecessary reconstructions are done. 

\begin{figure}[tb]
  \centering
\centering
\includegraphics{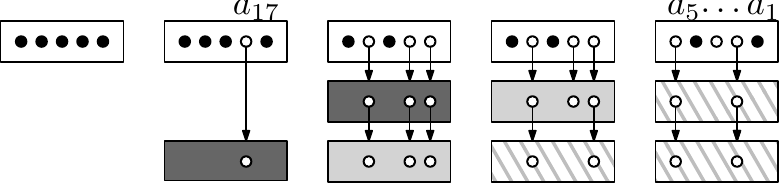}
\caption {Push operation: the top row has the $25$ input values partitioned into blocks of size $5$ (white points indicate values that will be pushed during the execution of the algorithm; black points are those that will be discarded). The middle and bottom rows show the situation of the compressed stack before and after $a_{17}$ has been pushed into the stack. Block $\Fir$ is depicted in dark gray, $\Nex$ in light gray, and the remaining compressed blocks with a diagonal stripe pattern.}
\label{fig:push}
\end{figure}

We now explain how to handle the push and pop operations with the compressed stack data structure. Recall that $\Fir$ and $\Nex$ are the only stack blocks to be stored explicitly. There are two cases to consider whenever a value $a$ is pushed: if $a$ belongs to the same block as the current top of the stack, i.e.  $\BFir$, it must go into $\Fir$ (and thus can be added normally to the stack in constant time). Otherwise $\alg$ has pushed an element from a new block, hence we must update $\Fir$ and $\Nex$ as follows: $\Fir$ will be a new portion of the stack that will only contain $a$. If the old $\Fir$ is empty, no change will happen to $\Nex$. Otherwise,  $\Nex$ becomes the previous $\Fir$ and we must compress the former $\Nex$ (see Fig.~\ref{fig:push}). All these operations can be done in constant space by smartly reusing pointers.

The pop operation is similar: as long as $\Fir$ contains at least one element, the pop is executed as usual. If $\Fir$ is empty, we pop values from $\Nex$ instead. The only special situation happens when after a pop operation the block $\Nex$ becomes empty. In this case, we must update block $\Nex$ (note that $\Fir$ remains empty, unchanged). For this purpose, we pick the topmost compressed block from the stack (if any) and reconstruct it in full. Recall that, by Lemma \ref{lem_reconst}, we can do so in $O(n/p)=O(\sqrt n)$ time using $O(n/p)=O(\sqrt{n})$ variables. The reconstructed block becomes the new $\Nex$ (and $\Fir$ remains empty).

\begin{theorem}\label{thm_comp1}
The compressed stack technique can be used to transform \alg\ into an algorithm that runs in $O(n)$ time and uses $O(\sqrt n)$ variables. 
\end{theorem}
\begin{proof}
The general workings of \alg\ remain unchanged, hence the difference in the running time (if any) will be due to push and pop operations. In most cases these operations only need constant time. The only situation in which an operation takes more than constant time is when a pop is performed and $\Nex$ becomes empty. In this situation we must spend $O(n/p)=O(\sqrt{n})$ time to reconstruct another block from the stack.

We now show that \textsc{Reconstruct}  is not invoked many times. Recall that this procedure is only invoked after a pop operation in which $\Nex$ becomes empty. We charge the reconstruction cost to $\BNex$ and claim that this block can never cause another reconstruction: by the order property, from now on \alg\ can only push elements of $\BFir$ or from a new block. In particular, \alg\ will not push elements from $\BNex$ again. That is, no element of $\BNex$ is on the stack (since $\Nex$ became empty) nor new elements of the block will be afterwards pushed (and thus $\BNex$ cannot cause another reconstruction). Thus, we conclude that no block can be charged twice, which implies that at most $O(n/p)$ reconstructions are done. Since each reconstruction needs $O(p)$ time, the total time spent reconstructing blocks is bounded by $O(n)$. 


Regarding space use, at any point of the execution we keep a portion of the stack containing at most $2(n/p) = 2\sqrt{n}$ elements in explicit form. The remainder of the stack is compressed into at most $p-2$ blocks. The top two blocks need $O(n/p)$ space whereas the remaining blocks need $O(1)$ space each. Hence the space needed is 
$O( n/p+p)$,
 which equals $O(\sqrt n)$ if $p=\sqrt n$.
\end{proof}

\subsection{General case}\label{sec_gen}


\begin{figure}[tb]
\centering
\includegraphics[width=1\textwidth]{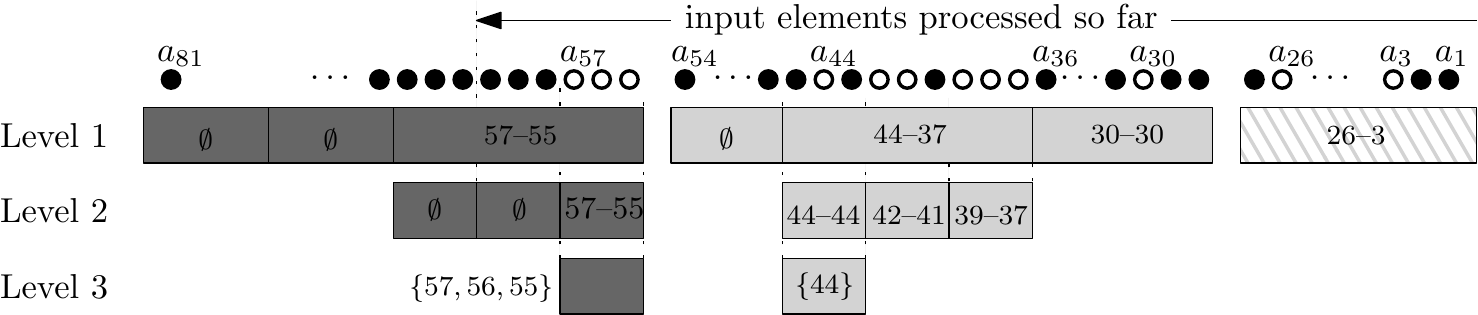}
\caption {A compressed stack for $n=81$, $p=3$ (thus $h=3$). The compression levels are depicted from top to bottom (for clarity, the size of each block is not proportional to the number of elements it contains). Color notation for points and blocks is as in Fig.~\ref{fig:push}. Compressed blocks contain the indices corresponding to the first and last element inside the block (or three pairs if the block is partially compressed); explicitly stored blocks contain a list of the pushed elements.}
\label{fig:CompressedStack}
\end{figure}

In the general case we partition the input into roughly $p$ blocks for any parameter $2\leq p\leq n$ (the exact value of $p$ will be determined by the user). At any level, we partition the input into $p$ blocks as follows:  let $k$ be the remainder of dividing $n$ by $p$. The first $k$ blocks have size $\lceil n/p\rceil$ whereas the remaining $p-k$ blocks have size $\lfloor n/p\rfloor$. From now on, for simplicity in the notation, we assume that $p$ divides $n$, and thus we can ignore all floor and ceiling functions. This assumption has no impact on the asymptotic behavior of the algorithm\footnote{Along the paper we use the Master Theorem to bound the running time of our algorithms. Strictly speaking, this theorem does not allow the use of floor and ceiling functions (thus, the bounds would not hold when $p$ does not divide $n$, and one of the blocks contains $\lceil n/p\rceil> n/p$ elements). This situation can be handled with a small case analysis, or alternatively by replacing the Master Theorem with the more general Akra-Bazzi Theorem~\cite{ab-otslre-98}.}. 

As the size of a block may be larger than the available workspace, instead of explicitly storing the top two non-empty blocks, we further subdivide them into $p$ sub-blocks. This process is then repeated for $h:=\log_p n-1$ levels until the last level, where the blocks are explicitly stored. 
 That is, in the general case we have three different levels of compression: a block is either stored (i) \emph{explicitly} if it is stored in full, (ii) \emph{compressed} if only the first and last elements of the block are stored, or (iii) \emph{partially compressed}, if the block is subdivided into $p$ smaller sub-blocks, and the first and last elements of each sub-block are stored. Analogously to the previous section, we store the context each time a block is created (i.e., whenever the first value of a block is pushed into the stack).

Blocks of different levels have different sizes, thus we must generalize the role of $\BFir$ and $\BNex$. In general, we define $\BFir_i$ as the block {\em in the $i$-th level} that contains the element that was pushed last into the stack. Likewise, $\BNex_i$ is the topmost block in the $i$-th level with elements in the stack (other than $\BFir_i$). We denote by $\Fir_i$ and $\Nex_i$ the portions of the stack of blocks $\BFir_i$ and $\BNex_i$, respectively. By definition of the blocks, we always have $\Fir_i\subseteq \Fir_{i-1}$. Moreover, either $\Nex_i \subseteq \Fir_{i-1}$ or $\Nex_{i} \subseteq \Nex_{i-1}$ hold (in addition, if the second case holds, we must also have $\Fir_{i}= \Fir_{i-1}$).

During the execution of the algorithm, the first (or highest) level of compression will contain $p$ blocks, each of which will have size $n/p$. We store $\Fir_1$ and $\Nex_1$ in partially compressed format, whereas the portion of the stack corresponding to all other blocks (if any) is stored in compressed format. That is, blocks $\Fir_1$ and $\Nex_1$  are further subdivided into $p$ sub-blocks of size $n/p^{2}$ each. 
In the $i$-th level of compression (for $1<  i< h $) $\Fir_i$ and $\Nex_i$ have size $n/p^{i}$ each, and are kept in partially compressed format. 
The sub-blocks corresponding to $\Fir_{i+1}$ and $\Nex_{i+1}$ are given to the lower level, which are again divided into sub-sub blocks, and so on. The process continues recursively until the $h$-th level, in which the subdivided blocks have size $n/p^h=n/p^{\log_p n-1}=p$. In this level the blocks are sufficiently small, hence we can explicitly store $\Fir_{h}$ and $\Nex_{h}$. 
%
See Fig.~\ref{fig:CompressedStack} for an illustration. 

\begin{lemma}\label{lem_space}
The compressed stack structure uses $O(p\log_p n)$ space.
\end{lemma}
\begin{proof}
At the first level of the stack we have $p$ blocks. The first two are partially-compressed and need $O(p)$ space each, whereas the remaining blocks are compressed, and need $O(1)$ space each. Since the topmost level can have at most $p$ blocks, the total amount of space needed at the first level is bounded by $O(p)$. 

At other levels of compression we only keep two partially-compressed blocks (or two explicitly stored blocks for the lowest level). Regardless of the level in which it belongs to, each partially-compressed block needs $O(p)$ space. Since the total number of levels is $h=\log_pn-1$, the algorithm will never use more than $O(p\log_pn)$ space to store the compressed stack.
\end{proof}

We now explain how to implement the push and pop operations so that the compressed stack's structure is always maintained, and the running time complexity is not highly affected.

\textbf{Push operation. }
A push can be treated in each level $i\leq h$ independently. First notice that by the way in which values of $\mathcal{I}$ are pushed, the new value $a$ either belongs to $\BFir_i$ or it is the first pushed element of a new block. In the first case, we register the change to $\Fir_i$ directly by updating the top value of $\Fir_i$ (or adding it to the list of elements if $i=h$). If the value to push does not belong to $\Fir_i$, then blocks $\Fir_i$ and $\Nex_i$ need to be updated. Essentially, the operation is the same as in Section~\ref{sec_lin}: $\Fir_i$ becomes a new block only containing $a$. Since we are creating a block, we also store the context of $a$ in the block. Also, if the former $\Fir_i$ was not empty, $\Nex_i$ becomes the previous $\Fir_i$. 

If $\Nex_i$ is updated, we should compress the old $\Nex_i$ as well. However, this is only necessary at the topmost level (i.e., when $i=1$). At other levels, the same block that should  be compressed is already stored in one level above. So, whenever necessary, we can query level $i-1$ and obtain the block in compressed format. For example, in Fig.~\ref{fig:CompressedStack}, level 2: if at some point we need to compress the partially compressed block $S_2$ (44-44, 42-41, 39-37), this is not actually needed since the same block is already compressed in level 1 (44-37). As in Section~\ref{sec_lin}, these operations can be done in constant time for a single level. 

\textbf{Pop operation.}
This operation starts at the bottommost level $h$, and it is then transmitted to levels above. Naturally, we must first remove the top element of $\Fir_h$ (unless $\Fir_h=\emptyset$ in which case we must pop from $\Nex_h$ instead). The easiest situation occurs when there are more elements from the same block in the stack. Recall that $\Fir_h$ and $\Nex_h$ are stored explicitly, thus we know which element of the stack will become the new top. Thus, we transmit the new top of the stack to levels above. In those levels, we need only update the top element of the corresponding sub-block and we are done.

A more complex situation happens when the pop operation empties either $\Fir_h$ or $\Nex_h$. In the former case, we transmit the information to a level above. In the higher level we mark the sub-block as empty and, if this results in an empty block, we again transmit the information to a level above, and so on. During this procedure several blocks $\Fir_i$ may become empty, but no block of type $\Nex_j$ will do so (since this would imply that $\Fir_i$ is included in a block $\Nex_{i-1}$, which cannot happen). Note that this is no problem, since in general we allow blocks $\Fir_i$ to be empty.

If block $\Fir_h$ is empty we must pop from $\Nex_h$ instead. Again, no more operations are needed unless $\Nex_h$ becomes empty. As before, we transmit the information to higher levels marking the sub-blocks as empty. We stop at a level $i$ in which block $\Nex_i$ is empty and in which $\Nex_{i-1}$ is not empty (or when $\Nex_1$ is empty). In order to preserve our invariant, we must invoke the reconstruction procedure to obtain the new blocks $\Nex_j$ for all $j\geq i$. Since $\Nex_{i-1}$ is not empty (even though $\Nex_{i}$ is), we can obtain the first and last elements of the block of $\Nex_i$ that are currently in the stack. If $i=1$ and we reached the highest level, we pick the first compressed block and reconstruct that one instead. In either case, the first and last elements of the block to reconstruct are always known. Note that it could also happen that at the highest level we have no more compressed elements. However, this corresponds to the case in which the stack contained a single element and thus, the stack is now empty.

\textbf{Block Reconstruction.}
This operation is invoked when the portion of a stack corresponding to a block $B$ needs to be reconstructed. 
 This operation receives: $(i)$ the first and last elements of $B$ that are in the stack (denoted $a_b$ and $a_t$, respectively), and $(ii)$ the context right after $a_b$ was pushed. Our aim is to obtain $\overline{B}$ in either explicit format (if $B$ is a block of the $h$-th level) or in partially-compressed format (otherwise). 

Block reconstruction is done in a similar way to Lemma~\ref{lem_reconst} for $O(\sqrt{n})$ variables: we execute \alg\ initialized with the given context information, starting with the next element of $\mathcal{I}$ after $a_b$. We process the elements of $\mathcal{I}$ one by one, stopping once we have processed $a_t$. During this re-execution, all stack elements are handled in an independent auxiliary stack $\mathcal{S_A}$. Initially, this stack only contains $a_b$, but whenever any element is pushed or popped the operation is executed on $\mathcal{S_A}$. 

Note that we cannot always store $\mathcal{S_A}$ explicitly, so we use a compressed stack. Let $j$ be the level of block $B$. The auxiliary stack is stored in a compressed stack of $h-j+1$ levels. That is, at the topmost element of $\mathcal{S_A}$ we only have one block of size $n/p^j$ (i.e., $\overline{B}$). This block is subdivided into $p$ sub-blocks, and so on. As usual, we store each block in partially-compressed format except at the lowermost level in which blocks have $p$ elements and are explicitly stored. Note that, in the particular case in which $j=h$ (i.e., $B$ is a block of the lowermost level), there is no compression. That is, we use a regular stack to store all elements.

Once we have finished processing $a_t$, the auxiliary stack $\mathcal{S_A}$ will contain $\overline{B}$ (in partially compressed format if $j<h$ or explicitly otherwise). Thus, we can return it and proceed with the usual execution of $\alg$. 

As an initial step, we disregard running time and just show correctness of the reconstruction. 

\begin{lemma}\label{lem_reconstgen}
\textsc{Reconstruct} procedure correctly reconstructs any block. Moreover, it will never use more than $O(p\log_p m)$ space, where $m$ is the number of elements between $a_b$ and $a_t$ in $\mathcal{I}$.
\end{lemma}
\begin{proof}
Proof of correctness is analogous to the one given in Lemma~\ref{lem_reconst}: the algorithm is initialized with $a_b$ and the context of $a_b$. In particular, the conditions evaluated during the execution of \textsc{Reconstruct} will be equal to those of \alg. By the invariant property, the situation of the stack of both executions after $a_t$ is treated will be the same. Thus, we conclude that once the \textsc{Reconstruct} procedure finishes, the elements that are in $\mathcal{S_A}$ will be identical to those of $\overline{B}$. 

Observe that during the execution of \textsc{Reconstruct} a block of $\mathcal{S_A}$ might become empty, which could trigger a nested \textsc{Reconstruct} procedure. This second procedure could trigger another reconstruction, and so on. However, we claim that this cannot go on indefinitely. Indeed, by the invariant property we know that $a_b$ is never popped from $\mathcal{S_A}$, which in particular implies that the block associated to $B$ cannot become empty. Since this is the only block of level $k$, we conclude that any block that is inductively reconstructed will have level higher than $B$. Since the number of levels is bounded, so is the total number of reconstructions. Thus, we conclude that \textsc{Reconstruct} will terminate. 

We now bound the space that this algorithm needs. Ignoring recursive calls, \textsc{Reconstruct} procedure basically replicates the workings of $\alg$ for a smaller input (of size $m$ to be precise). As argued before, the execution of a \textsc{Reconstruct} procedure could invoke another reconstruction. However, each time this happens the level of the block to reconstruct increases (equivalently, the size of the block to reconstruct is decreased by at least a factor of $p$). Thus, at most $\log_p m -1$ nested reconstructions can happen at the same time.

The space bottleneck of the algorithm is the space needed for the compressed stack of the \textsc{Reconstruct} procedure (and all of the nested reconstructions). We claim that the total space never exceeds $O(p\log_p m)$ space: at any point during the execution of  $\alg$ with the compressed stack, there can be up to $O(\log_p m)$ nested \textsc{Reconstruct} procedures---one per each level. Each of them uses an auxiliary compressed stack that, according to Lemma~\ref{lem_space}, needs $O(p\log m)=O(p \log_p \frac{n}{p^i})$ space, where $i$ denotes the level of the block to reconstruct. 

However, the key observation is that if at some point of the execution of $\alg$, at some level $i$, a \textsc{Reconstruct} is needed, then blocks $\Fir_i$ and $\Nex_i$ must be empty, and so will blocks $\Fir_j$ and $\Nex_j$ for all lower levels (for $i+1 \leq j \leq h$). Therefore, if at any instant of time, we need to reconstruct a block of level $i$, we know that the compressed stack at that moment of time will have $\Omega(p \log_p m_i)$ space that is not being used (corresponding to the space unused for the blocks $\Fir_j$ and $\Nex_j$ for $i\geq j\geq h$). Thus, the additional space needed for the auxiliary compressed stack used in the \textsc{Reconstruct} procedure can be charged to this unused space. Thus, we conclude that a call to \textsc{Reconstruct} does not require any space other than the one already allocated to the original compressed stack.

\end{proof}

\begin{theorem}\label{theo_gen}
Any stack algorithm can be adapted so that, for any parameter $2\leq p \leq n$, it solves the same problem in $O(n^{1+\frac{1}{\log p}})$ time using $O(p\log_p n)$ variables.
\end{theorem}
\begin{proof}
Notice that we can always determine the top element of the stack in constant time. Hence the main workings of $\alg$ are unaffected. By Lemma~\ref{lem_space}, the compressed stack never goes above our allowed space. The other data structures use only $O(1)$ space, hence we only need to bound the space used in the recursion. By Lemma \ref{lem_reconstgen}, we know that the space needed by \textsc{Reconstruct} (and any necessary recursive reconstructions) is bounded by $O(p\log_p n)$. 

In order to complete the proof we need to show that the algorithm indeed runs in the allotted time. First consider the additional time spent by the push and pop operations in the compressed stack. Each push operation needs a constant number of operations per level. Since each input value is pushed at most once, the total time spent in all push operations is $O(nh)$. If we ignore the time spent in reconstructions, the pop procedure also does a constant number of operations at each level, thus the same bound is obtained. 

We now bound the running time spent in all reconstructions by using a charging scheme: let $T(m)$ be the time used by \textsc{Reconstruct} for a block of size $m$ (including the time that may be spent in recursively reconstructing blocks of smaller size). 
Clearly, for inputs of size $p$ or less the reconstruction can be done in linear time by using $\alg$ and a regular stack (that is, $T(p)=O(p)$).

For larger pieces of the input, we consider $\alg$ as an algorithm that reconstructs the $p$ blocks of highest level. That is, if we are allowed to use a regular stack in all levels, we would have the recurrence $T(n)=pT(\frac{n}{p})+O(p)$. However, since we use a compressed stack, additional reconstructions might be needed. Recall that these reconstructions are only invoked whenever $\Nex_i$ becomes empty (for some $i\leq h$). As in the proof of Theorem~\ref{thm_comp1}, we know that block $\BNex_i$ cannot cause another reconstruction (since it has been totally scanned), thus we charge the cost of the reconstruction to that block. 

That is, each block can be charged twice: once when it is first reconstructed by $\alg$, and a second time when that block generates an empty $\Nex_i$. We note that each reconstruction of a block of size $n/p$ can cause another reconstruction of smaller levels. However, the cost of those reconstructions will be included in the corresponding $T(n/p)$ term. Thus, whenever we use a compressed stack the recurrence becomes $T(n)=2pT(\frac{n}{p})+O(p)$. By the Master Theorem, the solution of this recurrence is $T(n)=O(n^{\log_{p}2p})=O(n^{\log{2p}/\log{p}})=O(n^{1+\frac{1}{\log p}})$. Thus, in total the running time of $\alg$ becomes $O(nh+n^{1+\frac{1}{\log p}})$. However, because $h=\log_p n- 1\in O(n^{\frac{1}{\log p}})$, the first term can be ignored since it will never dominate the running time. 
%
%
\end{proof}

\section{Faster algorithm for green stack algorithms}\label{sec_green}
In this section we consider the case in which \alg\ is green. Recall that the condition for an algorithm to be green is to have a \lp\ operation. The intuition behind this operation is as follows: imagine that \alg\ is treating value $a_c$ which would pop one or more elements and empty $\Nex_i$ for some $i<h$. In this case, we would need to rescan a large part of the input to reconstruct the next block when we actually only need to obtain the top $2p$ elements that are in the stack (i.e., since we only need to store the next block in partially-compressed format). Moreover, we more or less know where these elements are located, thus we could ideally search locally.

Given an input value $a_q$, the \emph{upper neighbor} of $a_q$ when treating $a_c$ is the element $\unb{q} \in \mathcal{I}$  with the smallest index strictly larger than $q$ such that $\unb{q}$ is still in the stack after $a_c$ is processed by \alg, i.e., after all pops generated by $a_c$ have been performed.
Analogously, the \emph{lower neighbor} of $a_q$ when treating $a_c$  is the element $\lnb{q} \in \mathcal I$  with the largest index smaller or equal than $q$ such that $\lnb{q}$ is still in the stack after processing $a_c$.


For efficiency reasons, we restrict the procedure to look only within a given interval of the input. Thus, procedure \lp\ receives four parameters: $(i)$ the input value $a_c$ that generates the pop, $(ii)$ two input values $a_t,a_b\in \mathcal{I}$, such that $t>b$ and both $a_t$ and $a_b$ are in the stack immediately before processing $a_c$, and a query value $a_q$ such that $t\geq q \geq b$. \lp\ must return the pair $(\unb{q},\lnb{q})$, as well as their context. For consistency, whenever $\unb{q}$ (or $\lnb{q}$) does not lie between $a_t$ and $a_b$, procedure $\lp$ should return $\emptyset$ instead. 

In most cases, the elements that are in the stack satisfy some property (e.g., for the visibility problem, the points in the stack are those that are unobstructed by the the scanned part of the polygon). By asking for the stack neighbors of a query point, we are somehow asking for who is satisfying the invariant around $a_q$ (e.g., in visibility terms, we are asking what edge is visible from the fixed point in the direction of the query point). Note that this is the only non-transparent operation that the user must provide for the faster compressed stack method to work. 

In Section~\ref{sec_applis} we give several examples of green algorithms. In these cases, the \lp\ procedure consists of a combination of a binary search with the decision version of a problem in a fixed direction (i.e., what point is visible in a given direction?). These algorithms run in  $O(m\log m)$ time and use $O(s)$ variables, where $m=t-b$ (i.e., the number of input values between $a_t$ and $a_b$). Thus, from now on, for simplicity in the analysis we assume that the  \lp\ procedure runs in $O(m\log m)$ time and uses $O(1)$ space. 
Naturally, the result presented here adapt to slower/faster \lp\ procedures (or that use more space) by updating the running time/size of the workspace accordingly.

\subsection{Green stack algorithms in $O(\log n)$-sized workspaces}\label{subsec_ologn}

We first show how to use \lp\ so as to obtain a compressed stack structure that can run in small workspaces (say, $O(\log n)$ variables). Since we cannot store the whole compressed stack in memory, our aim is to use it until we run out of space.
 That is, apply the block partition strategy of the previous section, with $p=2$ for $h=s$ levels (recall that $s$ is our allowed workspace). The only difference in the data structure occurs at the lowermost level, where each block has size $n/2^s$. Although we would like to store the blocks of the lowest level explicitly, the size of a single block is too large to fit into memory (if $s\in o(\log n)$, we have $n/2^s\in \omega(s)$). Instead, we store $\Fir_s$ and $\Nex_s$ in compressed format. If we can use $O(\log n)$ variables, the lowermost level has constant size and thus can be explicitly stored. Recall that for each block in compressed or partially-compressed format we store the context of the first element that is pushed. Additionally, we store the context of the last element of each block as well.

\begin{lemma}\label{lem_spaceologn}
The compressed stack structure for $O(\log n)$-workspaces uses $O(s)$ space.
\end{lemma}
\begin{proof}
At each level we store a constant amount of information. Since we only have $O(s)$ levels, the compressed stack data structure does not use more than $O(s)$ space.
\end{proof}


Our aim is to execute $\alg$ as usual, using \lp\ to reconstruct the stack each time an element generates one or more pops. Thus, if a value does not generate any pop, it is treated exactly as in Section~\ref{sec_gen}. The only difference in the push operation is the fact that the last level may be in compressed format, and thus less information is stored.

Values of the input that generate no pop are handled as before, thus we need only consider the case in which an element $a_c\in\mathcal{I}$ generates one or more pops. First we identify which blocks in the stack are emptied by $a_c$. For this purpose, we execute \lp\ limiting the search within the first and last elements of $\Fir_h$ and with query value $q=\textsc{stack}.\textsc{top}(1)$ (i.e., the current top of the stack). 
If $\lnb{q}\neq \emptyset$ , we know that $a_c$ does not pop enough elements to empty $\Fir_h$. Otherwise, we know that $\Fir_h$ has become empty, thus we make another \lp\ query, this time searching within $\Nex_h$ and with the top element of $\Nex_h$ as the query element. If again there does not exist a lower neighbor, we can mark both $\Fir_h$ and $\Nex_h$ as empty, and proceed to the next non-empty block (i.e., the largest $i$ such that $\Fir_i$ is not empty, or the largest $j$ such that $\Nex_j$ is not empty in case all $\Fir_i$ blocks are empty). In general, the algorithm executes \lp\ procedure limiting the search within the smallest non-empty block of the stack with the topmost element of that sub-block as the query value. If the block is destroyed we empty it and proceed to the next one, and so on. This process ends when either $\Fir_1 \cup \Nex_{1}=\emptyset$ (i.e., $a_c$ completely empties the stack) or we reach a level $i$ in which either $\Fir_i$ or $\Nex_i$ is not empty. 



Once this cascading has finished, we obtain the lower neighbor of $a_q$ (i.e., the first element of the stack that is not popped by $a_c$). Since \lp\ provides us with the context of this element, we can proceed with $\alg$ as usual and determine whether or not $\alg$ would push $a_c$ into the stack. 

Before analyzing the running time of this approach, we introduce the concept of complexity of an element of $\mathcal{I}$. Given a point $a\in\mathcal{I}$, let $\mathcal{D}_a$ be the collection of blocks that $a$ would destroy in the execution of~$\alg$. For any block $B\in \mathcal{D}_a$, let $|B|$ denote its size (i.e., $|B|=n/2^i$ if $B$ is a block of the $i$-th level), and let $m_a=\sum_{B\in \mathcal{D}_a}|B|$ be the {\em complexity} of $a$.

\begin{lemma}\label{lem_popologn}
An element of the input $a_c\in\mathcal{I}$ is correctly treated in $O(\frac{n\log n}{2^s}+m_a\log (m_a))$ time. Moreover, it holds that $\sum_{a\in\mathcal{I}}m_a=O(sn)$.
\end{lemma}
\begin{proof}
Correctness of the algorithm is guaranteed by the \lp\ procedure: each time we invoke the procedure, we do so with a query element  $a_q$  that is in the stack (since it is the first element of some sub-block), for which we know that all other elements above $a_q$ have been already popped. In particular, the lower neighbor of $a_q$ will tell us the first element of the stack that is not popped by $a_c$. After that we proceed with \alg\ as usual, so correctness holds. 

We now show that the bound on the running time holds. For each element of $\mathcal{I}$ that generates a pop we always start by executing the \lp\ procedure with a block at the lowermost level. 
Since the size of this block is $n/2^s$, this operation runs in $O(\frac{n}{2^s}\log \frac{n}{2^s}) = O(\frac{n\log n}{2^s})$ time, and dominates the running time of treating $a_c$ unless both $\Fir_h$ and $\Nex_h$ are emptied. 

The second term can only dominate when $m_a$ is large (and thus, destroys $\Nex_h$ as well as other blocks). First observe that in this case we scan $O(m_a)$ elements of the input. Indeed, we invoke the \lp\ operation with portions of the input of increasing size, and only when the corresponding portions of the stack have been emptied. In particular, we stop as soon as a block is not emptied by the current element. Since the block size at most doubles between two consecutive levels, we will never scan more than $O(m_a)$ elements. In total, the running time of treating $a_c$ is bounded by $O(\frac{n}{2^s}\log n)+\sum_{B\in \mathcal{D}_a}O(|B|)\log O(|B|)=O(\frac{n\log n}{2^s}+m_a\log (m_a))$ as claimed. 

To prove that $\sum_{a\in\mathcal{I}}m_a=O(sn)$, note that a block (or sub-block) can only be destroyed once in the whole execution of $\alg$. Hence, the total complexity of all blocks of a fixed level is $n$. Since we only have $h=s$ levels, the claim follows. 
\end{proof}

By the above result, this approach correctly treats an element of the input. Before proceeding to the next value, we possibly have to reconstruct one (or more) blocks. Recall that whenever we need to reconstruct a block, we can obtain in constant time the two indices $a_t$ and $a_b$ corresponding to the first and last elements of $B$ that have been pushed into the stack. As before, these elements can be obtained because they are stored in one level above (or because they are fully compressed  at the highest level). 

If the workspace we are in allows us to use $O(\log n)$ variables, and the block to reconstruct is at the lowermost level, then the block has a constant number of elements. Thus, we can reconstruct it explicitly by simply executing $\alg$ on it. Otherwise, we must obtain block $B$ in partially compressed format (since no other block is stored explicitly). Moreover, since we fixed $p=2$, we only need to find the first and last elements of the two sub-blocks that were pushed into the stack. Rather than reconstructing the block recursively, we use another approach that we call \textsc{GreenReconstruct}. By construction, if  $a_t$ and $a_b$ belong to the same sub-block, we know that the other sub-block is empty (and we have all the information needed to reconstruct $B$). Otherwise, we simply invoke \lp\ delimiting the search between $a_t$ and $a_b$ where the query element is the last element of the first sub-block of $B$. By definition of this operation, we will obtain the two missing elements so as to reconstruct $B$ (see Figure~\ref{fig_query2}).

\begin{figure}[tb]
\centering
\includegraphics[width=.6\textwidth]{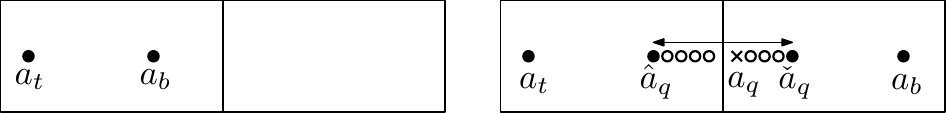}
\caption{Illustration of the \textsc{GreenReconstruct} procedure: (left) if $a_t$ and $a_b$ belong to the same sub-block, we know that the other one must be empty, thus $B$ can be reconstructed in constant time. Otherwise, a single execution of the \lp\ procedure with query $a_q$ as the last element of the second sub-block of $B$ will give us the other two elements needed to reconstruct $B$ (right).}
\label{fig_query2}
\end{figure}



\begin{lemma}\label{lem_reconstructgreen}
Procedure \textsc{GreenReconstruct} correctly reconstructs a block of $m$ elements in $O(m\log m)$ time. Moreover, the total time spent reconstructing blocks is bounded by $O(sn\log n)$.
\end{lemma}
\begin{proof}
Since we know that both $a_t$ and $a_b$ are in the stack when $a_c$ is processed (and we restrict the search within those two values), all of our \lp\ queries must return both an upper and lower neighbor (that is, the query cannot return $\emptyset$). Moreover, by the invariant property we know that any element that was in the stack in the moment that $a_t$ was pushed must remain in the stack when $a_c$ is processed. Thus, the reconstruction that we obtain through \lp\ will be the desired one. 





The running time of reconstructing a single block is dominated by the single execution of \lp\ operation, which takes $O(m\log m)$ time. In order to complete the proof we must show that not many reconstructions are needed. This is done with a charging scheme identical to the one used in Theorem~\ref{theo_gen}, and the analysis of Lemma~\ref{lem_popologn}: as in Section~\ref{sec_gen}, we only invoke $\textsc{GreenReconstruct}$ with a block of level $i$ when another block of the same level has been marked as empty. Hence, we charge the running time of reconstructing that block to the destroyed one. Since each block is only destroyed once, and no block can be charged twice, we will at most reconstruct $2^i$ blocks of level $i$. Each block has size $n/2^i$, hence we spend at most $2^i\times O(\frac{n}{2^i}\log \frac{n}{2^i})=O(n\log n)$ time reconstructing the elements of a single level. If we add this time among all $s$ levels, we obtain the claimed bound.
\end{proof}

We summarize the results of this section with the following statement.

\begin{theorem}\label{theo_timelogn}
Any green stack algorithm can be adapted so that it solves the same problem in $O(\frac{n^2\log n}{2^s}+n\log^2n)$ time using $O(s)$ variables (for any $s\in O(\log n)$).
\end{theorem}
\begin{proof}
Lemmas~\ref{lem_popologn} and ~\ref{lem_reconstructgreen} show that each element of $\mathcal{I}$ is properly treated, and that the invariant of the compressed stack is preserved at each step, respectively. Lemma~\ref{lem_spaceologn}  bounds the size the compressed stack data structure. Other than this structure, only a constant number of variables are used (even when executing the \lp\ procedure), hence we never exceed the allowed space.

Finally, we must argue that the algorithm runs in the specified the running time. The \textsc{Push} operation has not been modified, so it uses $O(1)$ space and time per level. Thus, the overall time spent in  \textsc{Push} operations is bounded by $O(ns)$. If we disregard the time needed for reconstructions, Lemma~\ref{lem_popologn} bounds the total time spent treating all elements of $\mathcal{I}$ by $O(\frac{n\log n}{2^s}+ns\log n)$. By the second claim of Lemma~\ref{lem_reconstructgreen}, the running time of executing all reconstructions is also bounded by $O(ns\log n)$.

Thus, in total we spend $O(ns+ \frac{n^2\log n}{2^s}+ns\log n)$ time and use $O(s)$ variables (for any $s\in O(\log n)$). In order to complete the proof we must distinguish between the $s\in o(\log n)$ and $s=\Theta(\log n)$ cases. Indeed, if $s\in o(\log n)$ we have $s\leq \log n/2$ and thus  $\frac{n^2\log n}{2^s} \geq \frac{n^2\log n}{2^{\log n/2}}= n^{1.5}\log n$. The other two terms are asymptotically smaller than $n\log^2n$, so the second term dominates. Whenever $s=\Theta(\log n)$ the third term dominates and is equal to $O(n\log^2n)$, thus the result follows.
%
\end{proof}

\subsection{Green algorithms for $\Omega(\log n)$-workspaces}\label{sec_hyb}
By combining both the standard and the green approach, we can obtain an improvement in the running time whenever $\Omega(\log n)$ variables are allowed. More precisely, we obtain a space-time trade-off between $O(n \log^2 n)$ time and $O(\log n)$ space, and $O(n\log^{1+\varepsilon} n)$ time and $O(n^\varepsilon)$ space for any $\varepsilon>0$. 

For this purpose, we introduce a {\em hybrid compressed stack} that combines both ideas. Specifically, we use the compressed stack technique for green algorithms for the first $h$ levels (where $h$ is a value to be determined later). That is, at the highest levels each block is partitioned into two, and so on. In particular, blocks at the $h$-th level have size $n/2^h$. For blocks at the $h$-th level or lower we switch over to the general algorithm of Section~\ref{sec_gen}: partition each block into $p$ sub-blocks (again, for a value of $p$ to be defined later). We stop after $\log_p (n/2^h) -1$ levels, in which the lowest blocks have $p$ elements and are explicitly stored. 

Blocks $\BFir_i$, $\BNex_i$, $\Fir_i$, and $\Nex_i$ are defined as usual. The handling of each block is identical to that of the corresponding section (i.e, if a block belongs to the $h$ topmost levels, then push, pop and reconstruct operations are executed as described in Section~\ref{subsec_ologn}; otherwise, we use the methods in Section~\ref{sec_gen}). 

\begin{lemma}\label{lem_runhyb}
The running time of executing $\alg$ with the hybrid compressed stack data structure is bounded by $O(nh\log n + 2^h (n/2^h)^{(1+1/\log p)})$.
\end{lemma}
\begin{proof}
Analysis is the same as in Theorems~\ref{theo_gen} and~\ref{theo_timelogn}, depending on the level. A push operation is handled in constant time per level, and since there are $h+\log_p (n/2^h)\leq \log n$ levels this cannot dominate the running time. Likewise, the time needed to handle a pop is also amortized in the number of destroyed blocks. 

The running time is dominated by the cost of the reconstruct operations. The first term of the expression comes from the higher levels, for which we use the \textsc{GreenReconstruct} operation: by Lemma~\ref{lem_reconstructgreen}, we spend $O(n\log n)$ per level (and there are $h$ levels). The second term arises from the lower levels, in which the recursive \textsc{Reconstruct} method is used. In this case, we simply apply Theorem~\ref{theo_gen} to each of the $2^h$ subproblems. Since each subproblem has size $n/2^h$, we obtain the claimed bound.
\end{proof}

In order to make the running time as small as possible, we choose the value of $p$ such that the two terms become equal: 
\begin{eqnarray}
nh\log n &=&2^h (n/2^h)^{(1+1/\log p)} \Leftrightarrow \\
h\log n&=&(n/2^h)^{(1/\log p)} \Leftrightarrow\\
\log h+\log\log n&=&\frac{\log (n/2^h)}{\log p} \Leftrightarrow\\
\log p&=&\frac{\log n - h}{\log h+\log\log n}
\end{eqnarray}

From Lemma~\ref{lem_runhyb}, we can see that shrinking $h$ by a constant factor only causes a constant multiplicative speed-up. In order to improve the result of Theorem~\ref{theo_gen}, we set $h = \log^a n$ for any constant $0<a<1$. This leads to a running time of $O(n \log^{(1+a)} n)$. 

\begin{theorem}\label{theo_hybrid}
Any green stack algorithm can be adapted so that it solves the same problem in $O(n \log^{(1+a)} n)$ time using $O(n^{\frac{1}{(1+a)\log \log n}}\log \log n)$ variables (for any fixed constant $0<a<1$).
\end{theorem}
\begin{proof}
The bound on the running time is given by Lemma~\ref{lem_runhyb}, thus we focus on the space that this approach uses. In the higher levels we use $O(h)$ space (recall that each level needs a constant amount of information by Lemma~\ref{lem_spaceologn}). For the lower levels, we use $O(p\log_p (n/2^h))=O(p(\log (n)-h)/\log p)$ in total (in here each level needs $O(p)$ space, see Lemma~\ref{lem_space}). 

Since we fixed the value of $p$ so that $\log p =\frac{\log n-h}{\log h+\log\log n} = \frac{\log n -h}{(1+a)\log \log n}$, we have $p=2^{\frac{\log n -h}{(1+a)\log \log n}}$. Moreover, we also fixed $h=\log^an \in o(\log n)$. Thus, we conclude that the space requirements of the lower part of the algorithm dominate, and is given by:

$$p(\log (n)-h)/\log p  = 
2^{\frac{\log n -h}{(1+a)\log \log n}}(1+a)\log \log n =
O(2^{\frac{\log n}{(1+a)\log \log n}}\log \log n) = O(n^{\frac{1}{(1+a)\log \log n}}\log \log n) $$
\end{proof}

We note that the space requirements of this approach are less than $n^\varepsilon$ (since $1/\log\log n$ is asymptotically smaller than any small fixed constant). Whenever $n^\varepsilon$ space is available, one should use Theorem~\ref{theo_gen} instead to obtain linear running time.

\section{Compressed stack for $k > 1$}\label{seclargek}
In the previous sections we assumed that $\alg$ only accesses the top element of the stack in one operation. In Section \ref{sec_applis} we introduce several algorithms that must look at the top two elements of the stack. One can also imagine algorithms that look even further down in the stack, hence in this section we generalize the compressed stack so that the top $k$ elements of the stack are always accessible (for some positive constant $k$). 

A simple way to allow access to the top $k$ elements of the stack, without modifying the workings of the compressed stack, is to keep the top $k-1$ elements in a separate {\em mini-stack}, while only storing the elements located further down in the compressed stack. 
%
We consider the mini-stack as part of the context, hence it will be stored every time the context is stored.
Note that since $k$ is a constant, the mini-stack has total constant size.

\textbf{Stack operations.}
Whenever the original stack contains $k-1$ or fewer elements, all of these elements will be kept explicitly in the mini-stack, and the compressed stack will be empty. Push and pop operations in which the stack size remains below $k$ will be handled in the mini-stack, and will not affect the compressed stack. 

If the mini-stack already has $k-1$ elements and a new element $a$ must be pushed, we push $a$ to the mini-stack, remove its {\em bottommost} element $a'$ from the mini-stack, and push $a'$ into the compressed stack. Similarly, whenever we must pop an element, we remove the top element of the mini-stack. In addition, we also pop one element from the compressed stack and add this element as the bottommost element of the mini-stack. 

The \textsc{Reconstruct} operation is mostly unaffected by this modification. The only change is that whenever we reconstruct a block, we must also recreate the situation of the mini-stack right after the bottommost element was pushed. Recall that this information is stored as part of the context, hence it can be safely retrieved. 

All of the handling operations of the mini-stack need constant time, and never create more than one operation in the compressed stack. Thus, we summarize the idea of this Section with the following result:
\begin{lemma} 
By combining a compressed stack with a mini-stack, the top $k$ elements of the stack can be accessed and stored, without affecting the time and space complexity of the compressed stack technique. 
\end{lemma} 

\section{Applications}\label{sec_applis} 

In this section we show how our technique can be applied to several well-known geometric problems. For each problem we present an existing algorithm that is a (green) stack algorithm, where our technique can be applied to produce a space-time trade-off.

\subsection{Convex hull of a simple polygon}\label{sec_hull}

Computing convex hulls is a fundamental problem in computational geometry, used as an intermediate step to solve many other geometric problems. 
For the particular case of a simple polygon $\Poly$ (Fig.~\ref{fig:applications}(a)), there exist several algorithms in the literature  that compute the convex hull of $\Poly$ in linear time (see the survey by Aloupis~\cite{a-hltchasp}). 

Among others, we highlight the method of Lee~\cite{l-ofchsp-83}; this algorithm walks along the boundary of the given polygon (say, in counterclockwise order).  At each step, it looks  at the top two elements of the stack when determining if the current vertex makes a left turn, and thus must be pushed into the stack. It is straightforward to verify that indeed, this algorithm is a stack algorithm (with $k=2$). For ease of exposition we assume that the vertices of $\Poly$ are in general position: no three vertices are colinear. Before showing that it is also green, we first introduce a technical operation:

\begin{lemma}\label{lem_bridge}
Let $\mathcal{C}=(p_1,\ldots, p_n)$ be a simple polygonal chain with its vertices in general position such that $p_1$ and $p_n$ are vertices of $CH(\mathcal{C})$. Let $r$ be a ray emanating from a point $q\in p_1p_n$ that crosses $\mathcal{C}$. Then we can find an edge of $CH(\mathcal{C})$ that intersects the ray $r$ in $O(n)$ time using $O(1)$ variables.
\end{lemma}
\begin{proof}
This result was shown by Pilz~\cite{p-pc-13}. In the following we give a proof for completeness. Without loss of generality, we assume that segment $p_1p_n$ is horizontal, and that the angle that ray $r$ forms with the positive $x$-axis is in $(0,\pi)$. Our algorithm will search among candidate edges that cross the given ray $r$, starting from an edge that may not be in $CH( \mathcal{C})$, and walking on $\cal C$ in opposite directions, in order to find an edge that is guaranteed to be in the convex hull. Given a candidate edge $p_u p_v$, we say that a point $p_k \in \cal C$ \emph{dominates} $p_u p_v$ if $p_u p_v$ is not part of $CH(p_1, p_n, p_u, p_v, p_k)$ (i.e. $p_u p_v$ is not an edge of the convex hull of those five points). Note that  $p_u p_v \in CH(\mathcal{C})$ if and only if  $p_u p_v$ is not dominated by any other vertex of $\cal C$. 



The first candidate edge is the segment $e_0$ of $\mathcal{C}$ that crosses $r$ furthest away from $q$. Naturally, $e_0$ can be found in linear time by walking once from $p_1$ to $p_n$. In general, let $p_u p_v$ be the current candidate edge (for some $u<v$).
We keep the following invariant: $p_u p_v$ crosses $r$ and no vertex of $\cal C$ between $p_u$ and $p_v$ can dominate $p_u p_v$. 
Thus any vertex dominating $p_u p_v$ must be in one of the chains $\mathcal{C}_u=(p_1,\ldots, p_u)$ and $\mathcal{C}_v=(p_v,\ldots, p_n)$. 
The algorithm searches for a dominating vertex by performing a tandem walk (i.e., first it checks vertex $p_{u-1}$ in $C_u$, then $p_{v+1}$ in $C_v$, $p_{u-2}$ in $C_u$, and so on). 
The search finishes when a first dominating vertex $p_k$ is found, or when all vertices in both chains have been checked.


\begin{figure}[tb]
\centering
\includegraphics[width=.5\textwidth]{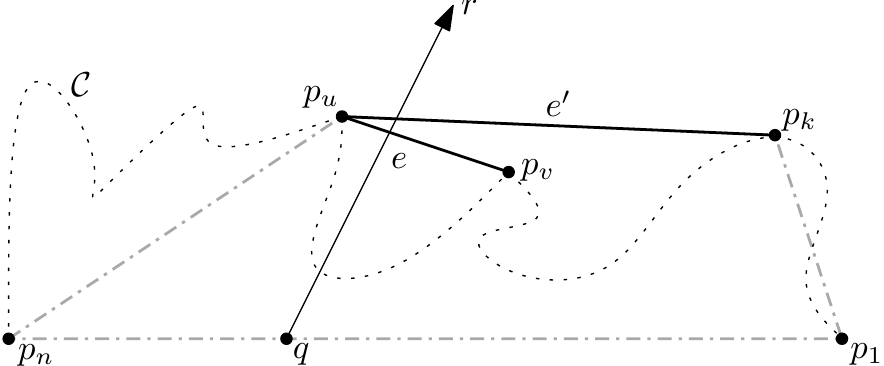}
\caption {
If $p_k$ dominates edge $p_u p_v$, then there must be an edge of $CH(p_1, p_n, p_u, p_v, p_k)$---drawn with a dash-dot pattern---that crosses $r$ and is incident to $p_k$. In the figure, that edge is $e'=p_u p_k$.
}
\label{fig_bridge}
\end{figure}

Checking a point $p_k$ of $\mathcal{C}_v$ or $\mathcal{C}_u$ consists in computing $CH(p_1, p_n, p_u, p_v, p_k)$, and determining whether $p_u p_v$ is part of it.
If so, then $p_u p_v$ remains as candidate edge. Otherwise there must be an edge in $CH(p_1, p_n, p_u, p_v, p_k)$ that has $p_k$ as endpoint and crosses $r$;
that edge becomes the new candidate edge, see Figure~\ref{fig_bridge}. 
 The algorithm finishes when both chains are empty, and thus the candidate edge $p_u p_v$ belongs to the convex hull, since no other point in $	\cal C$ can dominate it. 

Regarding the running time, we restart the search every time a new candidate has been found, and the cost of the search can be charged to the portion of the chain that has been discarded. Since we walk on both chains in parallel, the amount of processed vertices will be proportional to the size of the discarded chain. The space needed by the algorithm is $O(1)$, thus the lemma follows.
\end{proof}
 
The above operation allows us to find a convex hull edge in a specified direction in time proportional to the number of vertices in the chain. Since an edge gives us two elements that will be consecutive in the stack (i.e, the two endpoints), this operation can also be used to find the predecessor (or successor) of an element in the stack (by querying with a ray passing arbitrarily close to the element). In Lee's algorithm, the context of a point in the stack consists in its predecessor in the stack, which can be found in linear time using Lemma~\ref{lem_bridge} as well. 

We now use this operation to prove that indeed Lee's algorithm is green.

\begin{lemma}\label{lem_leefully}
Lee's algorithm for computing the convex hull of a simple polygon~\cite{l-ofchsp-83} is green. Moreover, the expected running time of \lp$(a_c,a_t,a_b,a_q)$ is $O(m\log m)$ time, and will never use more than $O(1)$ variables, where $m=t-b$.
\end{lemma}
\begin{proof}
The main invariant of Lee's algorithm is that at any instant of time the vertices that are in the stack are those that belong to the convex hull of the scanned portion of the input. Since $k=2$, and we are given the context of both $a_t$ and $a_b$, we can check the elements in the stack below those two elements.  
Let $a_t'$ and $a_b'$ be the predecessors of $a_t$ and $a_b$, respectively. Observe that the vertices in the stack between $a_t$ and $a_b$ (if any) must form segments whose slopes are between those of $e_t=a_ta'_t$ and $e_b=a^{ }_{b}a_{b}'$. 
Our aim is to find two indices $u^*,v^*$ such that $u^*\leq q < v^*$, and that the segment $a_{u^*}a_{v^*}$ belongs to the convex hull of the given chain. By the invariant of Lee's algorithm, these two elements must be consecutive in the stack, and thus form the neighborhood of $a_q$. 

In order to find these two vertices, we virtually rotate the problem instance so that the segment $a_ta'_b$ is horizontal, and edges $e_t$ and $e_b$ are above $a_ta'_b$. Let $\mathcal{R}_{b,t}$ be the (possibly unbounded) region that is below the line through $e_t$, above the line through $e_b$, and above the line through $a_ta'_b$ (gray region in Figure~\ref{fig_convexhull}). Also, let $m$ be the midpoint in the segment $a_ta'_b$. We start by scanning the input and counting how many vertices between $a_t$ and $a_b$ lie in $\mathcal{R}_{b,t}$. If only a constant number of vertices belong to this region the problem can be solved easily: explicitly store these points, compute the convex hull explicitly, and look for the indices $u$ and $v$. Since only vertices in this region can contribute to the edge we look for, the obtained solution will be correct.

However, in most cases, $\mathcal{R}_{b,t}$ will contain many points of the input. In this case, we query with a ray emanating from $m$ towards a randomly selected point in $\mathcal{R}_{b,t}$. By Lemma~\ref{lem_bridge}, we obtain a convex hull edge that intersects the ray. Let $a_u$ and $a_v$ be the endpoints of the obtained edge. If $u\leq q < v$ then we are done (since we have found the answer to our query). Otherwise, we discard all the portion of the chain between $a_t$ and $a_v$ or between $a_u$ and $a_b$, and continue the search with the new chain endpoints. We repeat this process until either we have found $\hat{a}_q$ and $\check{a}_q$, or we have two endpoints whose associated region contains a constant number of points (and the neighbors of $a_q$ can also be found as described above). 

\begin{figure}[tb]
\centering
\includegraphics[width=.5\textwidth]{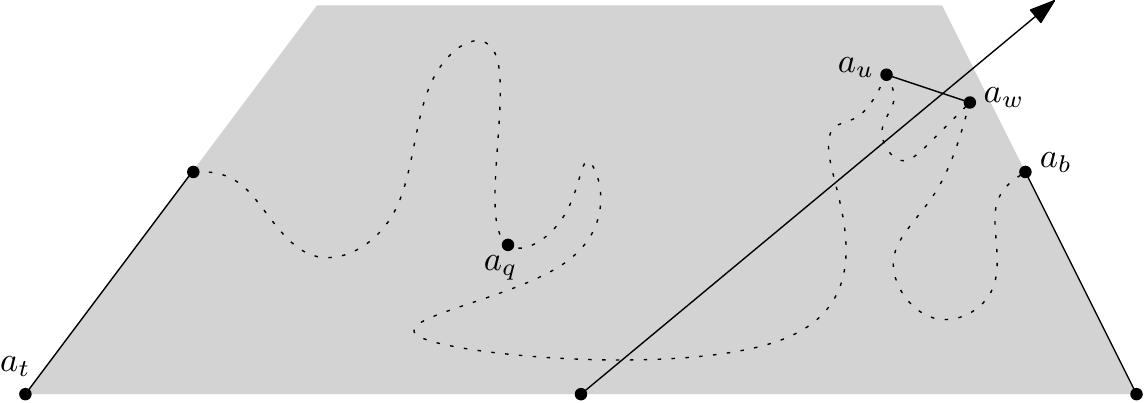}
\caption {Implementation of \lp\ operation for the convex hull. If $\mathcal{R}_{b,t}$ (gray in the figure) contains few vertices, we can solve the problem with any brute force algorithmm. Otherwise, we query a ray to obtain a convex hull edge, and obtain a subproblem with fewer points.}
\label{fig_convexhull}
\end{figure}

Thus, in order to complete the proof, we must show that show that we do not execute too many queries. Recall that the direction was determined by a point in $\mathcal{R}_{b,t}$ selected at random. Consider first the case in which we sort the points of the input radially when viewed from $m$ and pick the one whose angle forms the median angle. In this case, the number of points in region $\mathcal{R}_{b,t}$ would halve each step. The same holds, even if instead of a median we use an approximation (say, a point whose rank is between $1/3$ and $2/3$): at each step, a constant fraction of the points will be discarded.

Unfortunately, computing the median (or even an approximation) in $O(1)$-workspaces is too expensive. Thus, we pick a random direction instead. The probability that the rank of a randomly selected point is between $n/3$ and $2n/3$ is $1/3$ (i.e., a constant). If each of the choices is done independently at random, we are effectively doing Bernoulli trials whose probability is a constant. Thus, after a constant number of steps we will have selected an element whose angular rank is between $1/3$ and $2/3$ with high probability, and a constant fraction of the input will be pruned. In particular, we conclude that after $O(\log n)$ steps only a constant number of elements will remain. Since each step takes linear time, the bound holds.
\end{proof}


\begin{theorem}\label{the_hull}
The convex hull of a simple polygon can be reported in $O(n\log_p n)$ time using $O(p\log_p n)$ additional variables (for any $2\leq p\leq n$), $O(n\log^{1+a} n)$ time using $O(n^{\frac{1}{(1+a)\log\log n}}\log\log n)$ additional variables (for any $0<a<1$), or $O(n^2\log n/2^s+n\log^2n)$ time using $O(s)$ additional variables (for any $s\in O(\log n$)).
\end{theorem}

\subsection{Triangulation of a monotone polygon}
A simple polygon is called \emph{monotone} with respect to a line $\ell$ if for any line $\ell'$ perpendicular to $\ell$, the intersection of $\ell'$ and the polygon is connected. In our context, the goal is to report the diagonal edges of a triangulation of a given monotone polygon (see Fig.~\ref{fig:applications} (b)). Monotone polygons are a well-studied class of polygons because they are easier to handle than general polygons, can be used to model (polygonal) function graphs, and often can be used as stepping stones to solve problems on simple polygons (after subdividing them into monotone pieces). It is well-known that a monotone polygon can be triangulated in linear time using linear space\ShoLong{~\cite{gjpt-tsp-78}}{~\cite{gjpt-tsp-78,o-cgc-94}}. \ShoLong{}{We note that there also exists a linear-time algorithm that triangulates any simple polygon~\cite{c-tsplt-91} (that is, $\Poly$ need not be monotone), but that algorithm does not follow the scheme of Algorithm \ref{alg:scheme}, hence our approach cannot be directly used. However, our technique can be applied to a well-known algorithm for triangulating a monotone polygon due to Garey {\em et al.}~\cite{gjpt-tsp-78}. For simplicity, we present it below for $x$-monotone polygons. An $x$-monotone polygon is defined by two chains: the top chain and the bottom chain, which connect the leftmost vertex to the rightmost one.} 

\begin{lemma}\label{lem_garfully}
Garey {\em et al.}'s algorithm for triangulating a monotone polygon~\cite{gjpt-tsp-78} is green. Moreover, \lp$(a_c,a_t,a_b,a_q)$ can be made to run in expected $O(m\log m)$ time using $O(1)$ variables, where $m=t-b$.
\end{lemma}
\begin{proof}
Our technique can be applied to a well-known algorithm for triangulating a monotone polygon due to Garey {\em et al.}~\cite{gjpt-tsp-78}. For simplicity, we present it here for $x$-monotone polygons. An $x$-monotone polygon is defined by two chains: the top chain and the bottom chain, which connect the leftmost vertex to the rightmost one. The triangulation algorithm of Garey {\em et al.}~\cite{gjpt-tsp-78} consists in walking from left to right on both chains at the same time, drawing diagonals whenever possible, and maintaining a stack with vertices that have been processed but that are still missing some diagonal. 

At any instant of time, the vertices of the stack are a subchain of either the upper chain making left turns, or a subchain of the lower chain making right turns. In particular, at any time the points in the stack are consecutive vertices in either the lower envelope of the upper chain or the upper envelope of the lower chain. As the algorithm proceeds, it processes each vertex in order of $x$-coordinate. In a typical step, when a vertex $v$ is handled, the algorithm draws as many diagonals as possible from $v$ to other vertices in the stack, popping them, and finally adding $v$ to the stack. Note that, unlike in other results of this section, the stack contains the elements that have \emph{not} been triangulated yet. In all other aspects it follows the scheme of Algorithm \ref{alg:scheme}, hence it is a stack algorithm. The invariant of the elements in the stack is almost identical to the one for the convex hull problem (i.e., they form the upper or lower hull of one of the two chains), and thus the \lp\ operation is identical to the one given in Lemma~\ref{lem_leefully}.

%
\end{proof}

\begin{theorem}\label{theo_triang}
A triangulation of a monotone polygon of $n$ vertices can be reported in $O(n\log_p n)$ time using $O(p\log_p n)$ additional variables (for any $2\leq p\leq n$), $O(n\log^{1+a} n)$ time using $O(n^{\frac{1}{(1+a)\log\log n}}\log\log n)$ additional variables (for any $0<a<1$), or $O(n^2\log n/2^s+n\log^2n)$ time using $O(s)$ additional variables (for any $s\in O(\log n$)).
\end{theorem}


\vspace{-.1in}
\subsection{The shortest path between two points in a monotone polygon}
Shortest path computation is another fundamental problem in computational geometry with  many variations, especially queries restricted within a bounded region (see~\cite{m-spn-04} for a survey). Given a polygon $\Poly$, and two points $p,q\in\Poly$, their {\em geodesic} is defined as the shortest path that connects $p$ and $q$ among all the paths that stay within $\Poly$ (Fig.~\ref{fig:applications}(c)). It is easy to verify that, whenever $\Poly$ is a simple polygon,  the geodesic always exists and is unique. The length of that path is called the \textit{geodesic distance}. 

Asano {\em et al.}~\cite{amrw-cwagp-10,abbkmrs-mcasp-11} gave an $O(n^2/s)$ algorithm for solving this problem in $O(s)$-workspaces, provided that we allow an $O(n^2)$-time preprocessing. This preprocessing phase essentially consists in repeatedly triangulating $\Poly$, and storing $O(s)$ edges that partition $\Poly$ into $O(s)$ subpieces of size $O(n/s)$ each. 
Theorem~\ref{theo_triang} allows us to remove the preprocessing overhead of Asano {\em et al.} when $\Poly$ is a monotone polygon and there is enough space available.

\begin{theorem}\label{thm_shortest}
Given a monotone polygon $\Poly$ of size $n$ and points $p,q\in \Poly$, we can compute the geodesic that connects them in $O(n^2/s)$-time in an $O(s)$-workspace, for any $s$ such that $2 \log \log n \leq s < n$.
\end{theorem}
\begin{proof}
Recall that the algorithm of Asano {\em et al.}~\cite{abbkmrs-mcasp-11} has two phases: a triangulation part (that runs in $O(n^2)$ time) and a navigation part (that runs in $O(n^2/s)$ time). 
If we replace the triangulation procedure given by Asano {\em et al.} in \cite{abbkmrs-mcasp-11} by the algorithm given in Theorem~\ref{theo_triang}, the running time of the first phase changes from $O(n^2)$ to $O(n^2 \log n / 2^s)$, leading to an overall running time of $O(n^2 \log n / 2^s + n^2/s)$.
If $s \geq 2 \log \log n$, then $O(n^2 \log n / 2^s + n^2/s)=O(n^2/s)$.\footnote{The value that makes the two terms equal is between $\log \log n + \log \log \log n$ and $\log \log n + \log \log \log n+1$. However, the exact value is not relevant in our context. Thus, for simplicity we upper-bound it by $2 \log \log n$.}


Since the running time is now dominated by the navigation algorithm, we can move the preprocessing part of the algorithm to the query itself. The running time would become $O(n^2/2^s+n^2/s)=O(n^2/s)$. 
\end{proof}


\vspace{-.1in}
\subsection{Optimal 1-dimensional pyramid}
A vector $\phi = (y_1,\ldots,y_n)$ is called \emph{unimodal} if
$y_1 \le y_2 \le \cdots y_k$ and $y_k \ge y_{k+1} \ge \cdots y_n$
for some $1 \le k \le n$.
The $1$-D optimal pyramid problem~\cite{cst-ltaacspc-06} is defined as follows.
Given an $n$-dimensional vector $f = (x_1,\ldots,x_n)$, find  a unimodal vector $\phi = (y_1,\ldots,y_n)$
that minimizes the squared $L_2$-distance $|| f - \phi ||^2
= \sum_{i=1}^n (x_i - y_i)^2$ (Fig.~\ref{fig:applications}(d)). This problem has several applications in the fields of computer vision~\cite{b-uqsqqsrkrumm-03} and data mining~\cite{fmmt-dmotdar-01,mfmt-iedtrrs-97}. Although the linear-time algorithm of Chun {\em et al.}~\cite{cst-ltaacspc-06} does not exactly fit into our scheme, it can be modified so that our approach can be used as well.

\begin{theorem}\label{theo_pyr}
The $1$-D optimal pyramid for an $n$-dimensional vector can be computed in $O(n \log_p n)$ time using $O(p \log_p n)$ additional variables (for any parameter $2 \le p \le n$).
\end{theorem}
\begin{proof}
Chun {\em et al.}~\cite{cst-ltaacspc-06}  showed that if the location of the peak is fixed in the $k$-th position, the optimal vector
$\phi$ is given by the lower hull $H_\ell(k)$ of $x_1,\ldots,x_k$ and the lower hull $H_r(k)$ of $x_{k+1},\ldots,x_n$. Thus, their approach is to compute and store the sum of $L_2$-distances $D_\ell(k) = \sum_{i=1}^k (x_i - y_i)^2$ and $D_r(k) = \sum_{i=k+1}^n (x_i - y_i)^2$, and returning the index $i$ that minimizes $D_\ell(i)+D_r(i)$. 

The algorithm of Chun {\em et al.} uses an array of size $\Theta(n)$ to store the values $D_\ell(i)$ and $D_r(i)$ for $i\leq n$. Thus, before using our compressed stack technique we must first modify their algorithm so that this extra array is not needed. The idea is to compute values $D_\ell(i)$ and $D_r(i)$ in an incremental fashion, while storing at any instant of time the index $k$ that minimizes $D_\ell(k)+D_r(k)$. 

We scan the points of the input from left to right, processing values one by one. It is straightforward to modify the compressed convex hull algorithm so that, instead of reporting the hull, we compute the values $D_\ell(i)$ or $D_r(i)$ (depending on which points the convex hull was computed). Since we are scanning the points from left to right, we can easily obtain  $D_\ell(i)$ from $D_\ell(i-1)$ without affecting the time or space complexities (i.e., insert the point $(x_i,y_i)$, and update the convex hull). However, the same is not true when computing $D_r(i)$ from $D_r(i-1)$, since we have to {\em rollback} the treatment of a vertex in the convex hull algorithm (i.e., in this case we have to {\em remove} point $(x_i,y_i)$).

To achieve this, we use two stacks. The main one is used by the convex hull algorithm, thus at any instant of time it contains the elements that form the upper envelope $H_\ell(i)$. The secondary one contains all the points that were popped by the top element of the stack (recall that this element is denoted by $\textsc{pop}(1)$). Naturally, we keep both stacks in compressed format, hence the space bounds are asymptotically unaffected. 

Thus, we initialize the primary stack with $H_r(1)$, and the secondary one with the points that were popped by the top of $H_r(1)$. Note that $H_r(1)$ can be computed using Theorem~\ref{the_hull}. Whenever a rollback is needed, we pop the top element of the primary stack, and we push the elements in the secondary stack back into the primary one. Note that, although we do not have all of these points explicitly, we can obtain them invoking procedure \textsc{Reconstruct}. Also notice that, when we apply the reconstruction procedure, we can update the secondary stack with the input values that were popped by the new top of the stack.

We now show that the running time bounds hold. The initialization can be done in the specified time by directly using Theorem \ref{the_hull}. Observe that points may appear three times in the stack: once during the initialization phase, once in the secondary stack (when a vertex would remove them from the hull), and also they re-enter the primary stack a second time during a rollback (if they were pushed and popped in the primary stack at some point). A point $v$ is popped the second time from the primary stack only when we execute \textsc{Rollback}$(v)$. Hence, it will never be accessed again.

That is, a point is pushed (hence, popped) at most three times. As always, each push operation takes $O(h)$ time. Pop operations can also be handled with the same charging scheme as before, taking $O(h+n_a)$ time, where $n_a$ is the size of the destroyed block ($O(n/s^2+n_a)$ for $o(\log n)$-workspaces). Notice that a block can be charged three times, but this does not asymptotically increase the running time. 

Finally, observe that the main stack computes the convex hull, and since we have a green algorithm for this problem we could use the \lp\ operation. However, in order to do so we must also find the same operation for the secondary stack. We were unable to design an efficient \lp\ operation for this stack, thus we leave as an open problem determining if there exists a green stack algorithm for solving this problem.
\end{proof}

{\bf Remark.} Unlike in other applications, this algorithm loses the black box property. That is, \alg\ must know whether or not the stack is being handled in compressed format (and must act differently in each of the cases). This is caused by the need of the \textsc{Rollback} operation that is not needed in all other algorithms.

\subsection{Visibility profile in a simple polygon}
In the visibility profile (or polygon) problem we are given a simple polygon $\Poly$, and a point $q\in \Poly$ from where the visibility profile needs to be computed. A point $p\in \Poly$ is visible (with respect to $q$) if and only if $pq \subset \Poly$, where $pq$ denotes the segment connecting points $p$ and $q$. The set of points visible from $q$ is denoted by $\Vis$ and is called the \emph{visibility profile} (or \emph{visibility polygon}) of $q$ (see Fig. \ref{fig:applications}(e)). 
 Visibility computations arise naturally in many areas, such as computer graphics and geographic information systems, and have been widely studied in computational geometry. 
\ShoLong{}{

We refer the reader to the survey by O'Rourke \cite{r-v-04} and the book by Ghosh~\cite{g-vap-07} for a review of the planar visibility literature in memory-unconstrained models.} Among several linear-time algorithms, we are interested in the method of Joe and Simpson~\cite{js-clvpa-87} since it can be easily shown that it is green.

\begin{lemma}\label{lem_jsfully}
Joe and Simpson's algorithm for computing the visibility profile~\cite{js-clvpa-87} is green.
\end{lemma}
\begin{proof}
This algorithms scans the vertices of the polygon one by one in counterclockwise order. 
When processing a new vertex $v$, a number of cases are considered, which essentially depend on whether $q$, $v$, and its previous or next vertex make a right or left turn. A condition based on these cases determines whether some vertices that up to now were considered visible are not, thus must be popped, and whether $v$ is visible and thus should be pushed. Since vertices are processed in an incremental order, and $O(1)$ additional variables are used, Joe and Simpson's method is a stack algorithm.

We note that not only vertices of $\Poly$ are pushed into the stack: let $v$ be a visible vertex that makes a left or right turn. Consider  the ray emanating from $q$ towards $v$ and let $e$ be the first edge of $\Poly$ that properly intersects with the ray. The intersection point between $e$ and the ray is called the {\em shadow} of $v$ and is the last visible point from $q$ in the direction of $v$. It is easy to see that the visibility region is a polygon whose vertices are vertices that were originally present in $\Poly$ or shadows of reflex vertices. 


We now show how to implement the \lp$(a_c,a_t,a_b,a_q)$ operation. Geometrically speaking, this operation must return two consecutive $a_u, a_v$ vertices of \Vis\ so that $u\leq q <v$. Our aim is analogous to the approach for the 
the following: if $a_q$ is visible, then $\check{a}_q=a_q$, and $\hat{a}$ is the next counterclockwise vertex of \Vis. Otherwise, we must return the edge of \Vis\ that is crossed by ray from $q$ towards $a_q$. Since we know that both $a_t$ and $a_b$ are in the stack in the moment that $a_c$ is processed, we can restrict our search within the polygonal chain from $a_b$ to $a_t$ (since no other vertex of the input can cross the segments $qa_t$ and $qa_b$). Our approach is analogous to the one given in Lemma~\ref{lem_leefully}: we first give an algorithm for finding a visibility edge in a given query direction, and then combine it with a randomized search for the good direction.

For a fixed direction, we can find the first visible edge using the \textsc{RayShooting} operation described in~\cite{bkls-cvpufv-11}: essentially walk along the given polygonal chain and select the segment $e$ that crosses the ray closest to $q$. Since there is no obstruction, the edge $e$ (or at least a portion) must be visible. The last/first visible points of $e$ will be the shadow of the reflex vertex with counterclockwise smallest/largest angle among those that are in the triangle defined by $e$ and $q$, respectively (or an endpoint if no such reflex vertex exists). More details of this operation can be seen in~\cite{bkls-cvpufv-11} (Lemma~2). As with the convex hull, this operation takes linear time. The same reasoning shows that after a constant number of  \textsc{RayShooting} operations, the size of the input will have been reduced by a constant fraction, thus at most $O(\log n)$ queries will be needed.

\end{proof}


\begin{theorem}\label{theo_visi}
The visibility profile of a point $q$ with respect to $\Poly$ can be reported in $O(n\log_p n)$ time using $O(p\log_p n)$ additional variables (for any $2\leq p\leq n$), $O(n\log^{1+a} n)$ time using $O(n^{\frac{1}{(1+a)\log\log n}}\log\log n)$ additional variables (for any $0<a<1$), or $O(n^2\log n/2^s+n\log^2n)$ time using $O(s)$ additional variables (for any $s\in O(\log n$)).
\end{theorem}

\ShoLong{\section{Conclusions} 
In this paper we have shown how to transform any stack algorithm so as to work in memory-constrained models, and presented several concrete applications where it can be applied. Moreover, for many applications the technique can be applied in a black box fashion without altering the specifics of the algorithm. In addition, since the technique is rather simple to implement, we believe it can be useful in practice.
A natural open problem is extending this approach to other data structures (e.g. trees), which 
would allow many other useful algorithms to work in memory-constrained workspaces.
}{\section{Conclusions}
In this paper we have shown how to transform any stack algorithm so as to work in memory-constrained models. The main benefit is the fact that the method can be used in a black box fashion without knowing the specifics of the algorithm. Surprisingly, the space-time trade-off is exponential for small workspaces (i.e., increasing the workspace by a constant will halve the running time of the algorithm), whereas the improvement in larger workspaces is smaller. Hence, it seems natural to use $\Theta(\log n)$ workspaces whenever possible.
 
We note that the problem is much simpler when we are allowed to rearrange the values of the input: it suffices to partition the input into three parts (stack, discarded values, and values not processed yet), and rearrange the input values as necessary. Indeed, this fact was already observed by Br\"onnimann and Chan~\cite{bc-seacchspllt-06} (for the problem of computing the convex hull of simple polygons) and De {\em et al.}~\cite{dmn-seavpsp-12} (for the visibility problem). However, this method does not fit in our constant workspace model. 

A natural open problem is to improve the running time of our approach, or find an equivalent of the \lp\ procedure with weaker requirements. Another interesting problem would be extending this approach to other data structures. Mainly, we only use the monotone, order, and invariant properties from stack algorithms. Thus, in principle our approach should extend to other data structures/algorithms in which the equivalent properties are satisfied: for example, we are confident that this approach can be extended to {\em deques} (a stack-like structure in which we can push and pop from either extreme). However, it would be more interesting if we could also compress trees or more complex data structures. Methods for compressing deques or trees would allow us to generalize the algorithms presented in Section~\ref{sec_applis} so as to compute the convex hull of a polygonal line (instead of a polygon) or triangulate a simple polygon (instead of a monotone one), respectively. 

}

\paragraph{Acknowledgments}
We thank the anonymous referee for their thorough research as well as for pointing us to~\cite{cc-mpga-07}. We also thank Mikkel Abrahamsen for useful discussions and comments. M.K was partially supported by the Secretary for Universities and Research of the Ministry of Economy and Knowledge of the Government of Catalonia and the European Union.
R.S. was partially supported by FP7 Marie Curie Actions Individual Fellowship PIEF-GA-2009-251235 and by FCT through grant SFRH/BPD/88455/2012.
M.K and R.S. were also supported by projects  MINECO MTM2012-30951 and Gen. Cat. DGR2009SGR1040 and by ESF EUROCORES program EuroGIGA-ComPoSe IP04-MICINN project EUI-EURC-2011-4306.

\bibliographystyle{abbrv}
\bibliography{visi}

\end{document}